\documentclass[onecolumn,superscriptaddress,11pt]{revtex4-2}
\usepackage{graphicx} 
\usepackage{amsmath}
\usepackage[english]{babel}
\usepackage{mathtools}
\usepackage{braket}
\usepackage{amssymb}
\usepackage{hyperref}
\usepackage{bm}
\usepackage{comment}
\usepackage{amsthm}
\usepackage{color}

\newcommand{\ketbra}[2]{\ket{#1}\!\bra{#2}}
\newcommand{\tr}{\mathrm{Tr}}

\newtheorem{theorem}{Theorem}
\newtheorem{lemma}{Lemma}

\newtheorem{corollary}{Corollary}
\newtheorem{definition}{Definition}
\newtheorem{remark}{Remark}

\begin{document}
\title{Universal classical-quantum channel resolvability and private channel coding}
\author{Takaya Matsuura}
\email{takayamatsuura@gmail.com}
\affiliation{RIKEN Center for Quantum Computing (RQC), Hirosawa 2-1, Wako, Saitama 351-0198, Japan} 
\author{Masahito Hayashi}
\affiliation{School of Data Science, The Chinese University of Hong Kong,
Shenzhen, Longgang District, Shenzhen, 518172, China}
\affiliation{International Quantum Academy, Futian District, Shenzhen 518048, China}
\affiliation{Graduate School of Mathematics, Nagoya University,
Furo-cho, Chikusa-ku, Nagoya, 464-8602, Japan}
\author{Min-Hsiu Hsieh}
\affiliation{Hon Hai Research Institute, Taipei, Taiwan}

\begin{abstract}
    We address the problem of constructing fully universal private channel coding protocols for classical–quantum (c-q) channels. Previous work constructed universal decoding strategies, but the encoder relied on random coding, which prevents fully universal code construction. 
    In this work, we resolve this gap by identifying an explicit structural property of codebooks---namely, the spectral expansion of an associated Schreier graph---that guarantees universal channel resolvability. Our analysis reveals how this property can be related to channel resolvability through the theory of induced representation. 
    Our main technical result shows that when the transition matrix of a graph associated with a codebook has a large spectral gap, the channel output induced by uniformly sampled codewords is asymptotically indistinguishable from the target output distribution, independently of the channel. This establishes the first deterministic, channel-independent construction of resolvability codebooks.
    
    Building on this, we construct a fully universal private channel coding protocol by combining it with universal c-q channel coding based on the Schur-Weyl duality. With appropriate modifications to the requirements on codebooks of universal c-q channel coding, we show that our fully universal private channel coding achieves the known optimal rate. This work thus sheds new light on the expander property of a graph in the context of secure communication.
\end{abstract}
\maketitle

\section{Introduction}
Universal coding~\cite{Csiszar2004} aims to construct encoding and decoding strategies that do not depend on the detailed description of information sources or channels, but only on general properties such as the assumption of independent and identically distributed (i.i.d.) sources/channels and the size of alphabets. Even though universal encoders/decoders are not tailored to specific information sources/channels, they often achieve the same asymptotic performance as the optimal ones; they are characterized by the same information-theoretic quantities. Thus, universal coding not only has practical importance but also reveals what statistical properties are essential for a given information-theoretic task. Universal coding has been studied in quantum settings as well, and several protocols have been developed~\cite{Jozsa1998, Hayashi2002Quantum, Hayashi2009, Datta2010, Hayashi2010, Hayashi2022, Matsuura2025}.

For the case of universal classical–quantum (c–q) channel coding, Ref.~\cite{Hayashi2009} established a decoder that is completely independent of the channel. Building on this, Ref.~\cite{Datta2010} proposed a universal private channel coding~\cite{Devetak2005}, where a sender uses a c-q channel ${\cal X}\to {\cal D}({\cal H}_B\otimes {\cal H}_E)$ so that a receiver $B$ can reliably decode a classical message while an eavesdropper $E$ learns essentially nothing. 
However, their scheme relied on a random-coding argument for the encoder.  While this guarantees the existence of a good encoder, the actual choice of such an encoder may depend on the channel.  As a result, unless a sender and a receiver share common randomness, their protocol is not fully universal.  More fundamentally, the need for averaging in Ref.~\cite{Datta2010} arises from its reliance on the quantum soft-covering lemma~\cite{Ahlswede2002, Datta2010, Cheng2022error, Sen2025}. The quantum soft-covering lemma states that if a randomly chosen codebook is sufficiently large, then when a sender inputs a uniformly sampled codeword, the channel output is indistinguishable from an i.i.d.~state---but only on average, i.e., when averaged over codebooks. Therefore, in contrast to the classical case~\cite{Hayashi2016secure}, a fully universal private channel coding protocol has been missing in the quantum setting.

In this paper, we resolve this gap by constructing a sequence of channel-independent encoders that achieves both correctness and secrecy, thereby establishing a fully universal private channel coding protocol.
The key technical development is \emph{universal resolvability coding} for classical-quantum (c-q) channels. Resolvability coding is a way to construct codebooks such that the channel output, when fed with the uniform distribution over the codebook, is asymptotically indistinguishable from a target state---namely, the channel output induced by a uniform constant-composition input~\cite{Hayashi2024}. The soft-covering lemma guarantees this on average, i.e., when averaged over a random choice of codebooks, while the recent result~\cite{Sen2025} showed that a randomly chosen codebook achieves resolvability with high probability, i.e., without taking average over codebooks.
However, these results do not clarify what intrinsic property of codebooks ensures resolvability since the evaluation of the distinguishability between the channel output and the target state is not completely decoupled from the code construction in Ref.~\cite{Sen2025}.

Our contribution is to identify the exact property of (constant-composition) codebooks that achieve channel resolvability.  
For the classical case, the paper \cite[Section XII]{Hayashi2016secure} constructed a fully universal private channel coding. 
It focuses on the property that the channel can be decomposed into a sum of conditional types. However, when this method is translated into the c-q-channel case, this decomposition depends on the basis of the output system. Hence, this idea does not work with an output quantum system.
To show universal resolvability in the quantum setting, we completely change the strategy. We combine several mathematical tools for this: Schreier graph and its spectral gap, the representation theory of a permutation group, and the complex interpolation theory. 
We show that \emph{the expander property of a graph associated with a codebook} is the very property we need; a Schreier graph is associated with a codebook in a non-standard way (i.e., not a Tanner graph), and its spectral gap is directly connected to an upper bound on the state distance in terms of the 2-to-2 norm through representation-theoretic techniques.
Using this, one can upper-bound the usual distance between the output state of the channel using this codebook and the target state (i.e., the output state of the channel with a uniform constant-composition input) via the complex interpolation theory~\cite{Beigi2013, Dupuis2022, Beigi2023operator, Beigi2023}, which results in universal resolvability.
Recall that the Schur-Weyl duality has played a crucial role in the existing universal coding protocols~\cite{Jozsa1998, Hayashi2002Quantum, Hayashi2009, Datta2010, Hayashi2010, Hayashi2022, Matsuura2025}. While it also plays an important role in our work as well in the construction of universal c-q channel coding and universal private channel coding, we observe that a different representation-theoretic tool is necessary for universal resolvability coding, namely, the theory of induced representation. Especially, \emph{the universality of induced (unitary) representation} of a permutation group from its subgroup and the complete reducibility of a unitary representation of a permutation group are essential constituents of our proof.

Combining universal channel resolvability coding developed above with universal c-q channel coding~\cite{Hayashi2009}, we then aim at constructing universal c-q private channel coding~\cite{Datta2010}. Private channel coding is to reliably transmit a classical message to a receiver (correctness) while disturbing the information leaked through a side channel so that an eavesdropper cannot know the transmitted message (secrecy)~\cite{Devetak2005, Cai2004, Hayashi2015}. The main problem is to construct an encoder that satisfies both requirements for universal c-q channel resolvability coding to satisfy the secrecy and universal c-q channel coding to satisfy the correctness. We show the existence of such an encoding strategy by appropriately modifying the required conditions on the codebook.

The rest of the paper is organized as follows. In Sec.~\ref{sec:notations}, we describe the notations used in this work and the preliminaries for the results. In Sec.~\ref{sec:universal_c-q}, by employing an improved decoder, we improve the error exponent for c-q channel coding over that of the previous work~\cite{Hayashi2009}. We summarize the relevant property of codebooks that achieves the universal c-q channel coding, which is used in the later section. Section~\ref{sec:universal_resolvability} is our main result. We show that when a graph corresponding to a sufficiently large codebook possesses an appropriate expander property, then channel resolvability can be achieved with that codebook, independently of the channel.  Thus, the universal resolvability coding is possible as long as one knows the channel mutual information to determine the size of the codebook.  In Sec.~\ref{sec:universal_private}, we construct a universal private channel coding protocol, combining the universal c-q channel coding and the universal resolvability coding.

\section{Notations and preliminaries} \label{sec:notations}
This section summarizes mathematical tools to be used in the later sections. The following table shows which subsections are mainly used in the subsequent sections. Although preliminaries in \ref{sec:entropies} and \ref{sec:representation_theory} are not explicitly used in Sec.~\ref{sec:universal_private}, the result in Sec.~\ref{sec:universal_private} relies on Theorem~\ref{theo:universal_resolvability} in Sec.~\ref{sec:universal_resolvability}, which is based on them.
\begin{table}[h]
    \centering
    \begin{tabular}{lcl}
    \hline
        \ref{sec:Schur-Weyl}, \ref{sec:entropies} & $\Rightarrow$ & Section~\ref{sec:universal_c-q} \\
        \ref{sec:entropies}, \ref{sec:representation_theory}, \ref{sec:graph_theory} & $\Rightarrow$ & Section~\ref{sec:universal_resolvability} \\
        \ref{sec:Schur-Weyl}, \ref{sec:entropies}, (\ref{sec:representation_theory}, \ref{sec:graph_theory}) & $\Rightarrow$ & Section~\ref{sec:universal_private}\\
        \hline
    \end{tabular}
    \label{tab:placeholder}
\end{table}
\subsection{Schur-Weyl duality and type theory} \label{sec:Schur-Weyl}
Let ${\cal H} \coloneqq (\mathbb{C}^d)^{\otimes n}$.
Let $Y_n^d$ be the set of Young diagrams with $n$ boxes and at most $d$ rows.
The elements of $Y_n^d$ can be represented as $\bm{n}=(n_1, n_2, \ldots, n_d)$ with $n_1 \geq n_2 \geq \cdots n_d$ and $\sum_{i=1}^{d} n_i = n$.
Then, from the Schur-Weyl duality, ${\cal H}$ can be decomposed into
\begin{equation}
    {\cal H} = \bigoplus_{\bm{n}\in Y_n^d} {\cal U}_{\bm{n}} \otimes {\cal V}_{\bm{n}},
\end{equation}
where ${\cal U}_{\bm{n}}$ denotes the representation space of $\mathrm{SU}(d)$ and ${\cal V}_{\bm{n}}$ denotes that of the permutation group $S_n$.  For any $U\in \mathrm{SU}(d)$, $U^{\otimes n}$ can be decomposed into 
\begin{equation}
    U^{\otimes n} = \bigoplus_{\bm{n}\in Y_n^d} \pi_{\bm{n}}(U) \otimes I_{{\cal V}_{\bm{n}}},
\end{equation}
where $\pi_{\bm{n}}$ denotes the irreducible representation of $\mathrm{SU}(d)$ on ${\cal U}_{\bm{n}}$.  Similarly, any unitary representation $V_{s}$ of $s \in S_n$ can be decomposed into 
\begin{equation}
    V_{s} = \bigoplus_{\bm{n}\in Y_n^d} I_{{\cal U}_{\bm{n}}} \otimes \zeta_{\bm{n}}(s), \label{eq:permutation_schur_weyl}
\end{equation}
where $\zeta_{\bm{n}}$ denotes the irreducible representation of $S_n$ on ${\cal V}_{\bm{n}}$.  Any state that commutes with $U^{\otimes n}$ for all $U\in \mathrm{SU}(d)$ or commutes with $V_{s}$ for all $s\in S_n$ has the same block diagonal form in this Schur-Weyl basis from Schur's lemma.
Let $\Pi_{\bm{n}}$ be a projection onto the subspace labeled by $\bm{n}$.
Then, we define
\begin{align}
    \sigma_{\bm{n}} &\coloneqq \frac{\Pi_{\bm{n}}}{\dim ({\cal U}_{\bm{n}}\otimes {\cal V}_{\bm{n}})} , \label{eq:uniform_in_a_irreducible}\\
    \sigma_{U,n} &\coloneqq \sum_{\bm{n}\in Y_n^d} \frac{1}{|Y_n^d|} \sigma_{\bm{n}},
    \label{eq:uniform_state}
\end{align}
where $\sigma_{U,n}$ is called the universal symmetric state~\cite{Hayashi2009,Mosonyi2017,Hayashi2017}.  The state that commutes with both $U^{\otimes n}$ for all $U\in \mathrm{SU}(d)$ and $V_{s}$ for all $s\in S_n$ can be written as $\sum_{\bm{n}\in Y_n^d}p_{\bm{n}}\sigma_{\bm{n}}$, where $\{p_{\bm{n}}\}_{\bm{n}\in Y_n^d}$ is the probability distribution.  The state $\sigma_{\bm{n}}$ for any $\bm{n}\in Y_n^d$ commutes with any operator that has a block-diagonal form in the Schur-Weyl basis, and therefore so does $\sigma_{U, n}$.  In particular, the universal symmetric state $\sigma_{U, n}$ commutes with any operator in the form $O^{\otimes n}$, which will be used later. 
For any i.i.d.~state $\rho^{\otimes n}$, the following holds from its permutation symmetry:
\begin{equation}
    \Pi_{\bm{n}} \rho^{\otimes n} \Pi_{\bm{n}} = \rho_{{\cal U}_{\bm{n}}} \otimes \frac{I_{{\cal V}_{\bm{n}}}}{\dim {\cal V}_{\bm{n}}} \leq \frac{\Pi_{\bm{n}}}{\dim {\cal V}_{\bm{n}}}  \leq \dim {\cal U}_{\bm{n}} \sigma_{\bm{n}} ,\label{eq:direct_sum_decomp}
\end{equation}
where $\rho_{{\cal U}_{\bm{n}}}$ is an unnormalized density operator.
Therefore, we have
\begin{equation}
    \rho^{\otimes n} \leq \sum_{\bm{n}\in Y_n^d} \dim {\cal U}_{\bm{n}} \, \sigma_{\bm{n}} \leq \max_{\bm{n}}(\dim {\cal U}_{\bm{n}}) |Y_n^d| \sigma_{U,n}.
    \label{eq:bound_universal_symmetric}
\end{equation}
It is known that the following upper bounds hold~\cite{Itzykson1966,Hayashi2009, Mosonyi2017}:
\begin{equation}
    \dim {\cal U}_{\bm{n}} \leq (n+1)^{\frac{d(d-1)}{2}}
\end{equation}
and
\begin{equation}
    |Y_n^d| \leq (n+1)^{d-1}.
\end{equation}
Therefore, the coefficient of $\sigma_{U,n}$ in Eq.~\eqref{eq:bound_universal_symmetric} can be bounded from above by 
\begin{equation}
    \max_{\bm{n}}(\dim {\cal U}_{\bm{n}}) |Y_n^d|\leq (n+1)^{\frac{(d+2)(d-1)}{2}}. \label{eq:quantum_coefficient_bound}
\end{equation}

Let us now consider the string $\bm{x}\in{\cal X}^n$.  For ease of discussion, let ${\cal X}$ be the set of integers from $1$ to $|{\cal X}|$.  Define $P_{\bm{x}}$ as the type of the string $\bm{x}$, i.e., the probability distribution over ${\cal X}$ satisfying
\begin{equation}
    \forall y\in{\cal X},\qquad P_{\bm{x}}(y) = \frac{\left|\bigl\{i\in\{1,\ldots,n\}:x_i=y\bigr\}\right|}{n}.
\end{equation}
Let ${\cal P}_n$ be the set of types for length-$n$ strings.
It is known that the following holds:
\begin{equation}
    |{\cal P}_n| \leq (n+1)^{d-1}. \label{eq:number_of_types}
\end{equation}
For each $P\in{\cal P}_n$, let ${\cal T}_P$ be the set of length-$n$ strings with the type $P$, i.e.,
\begin{equation}
    {\cal T}_P \coloneqq \{\bm{x}\in{\cal X}^n:P_{\bm{x}}=P\}.
\end{equation}
For each set ${\cal T}_P$, we define a specific element $\bm{x}_P\in{\cal T}_P$ as
\begin{equation}
    \bm{x}_P \coloneqq (\underbrace{1, \ldots, 1}_{m_1}, \underbrace{2, \ldots, 2}_{m_2}, \ldots, \underbrace{|{\cal X}|, \ldots, |{\cal X}|}_{m_{|{\cal X}|}}), \label{eq:specific_element_in_type_P}
\end{equation}
where $m_i\coloneqq n P(i)$ for $i=1,\ldots,|{\cal X}|$.  Then, there exists a permutation $s_{\bm{x}_P\to\bm{x}}\in S_n$ for each $\bm{x}\in{\cal T}_P$ such that 
\begin{equation}
    \bm{x}=s_{\bm{x}_{P}\to\bm{x}}(\bm{x}_P), \label{eq:def_s_bm_x}
\end{equation}  
which is not unique.
For later use, we also define an isotropy group $S_{\bm{x}_P}\subset S_n$ for $\bm{x}_P$ as
\begin{align}
        S_{\bm{x}_P}&\coloneqq\{s\in S_n:s(\bm{x}_P)=\bm{x}_P\}. \label{eq:isotropy}
\end{align}

Given two strings $\bm{x}\in{\cal X}^n$ and $\bm{y}\in{\cal Y}^n$, let $P_{\bm{x},\bm{y}}$ be the joint type defined as
\begin{equation}
    \forall (u,v)\in{\cal X}\times{\cal Y}, \qquad P_{\bm{x},\bm{y}}(u,v)=\frac{|\{i\in\{1,\ldots,n\}:(x_i,y_i)=(u,v)\}|}{n}.
\end{equation}
Then, a conditional type $\bm{V}_{\bm{y}|\bm{x}}$ of $\bm{y}$ given $\bm{x}$ is defined to satisfy
\begin{equation}
    \forall (x,y)\in{\cal X}\times{\cal Y}, \qquad\bm{V}_{\bm{y}|\bm{x}}(y|x){P_{\bm{x}}}(x)= P_{\bm{x},\bm{y}}(x,y).
\end{equation}
Note that the conditional type is not uniquely determined for $x\in{\cal X}$ such that $P_{\bm{x}}(x)=0$.
Let ${\cal V}(\bm{x},{\cal Y})$ be the set of all the conditional types for $\bm{x}$.  Then, for any conditional type $\bm{V}\in{\cal V}(\bm{x},{\cal Y})$, we define the subset ${\cal T}_{\bm{V}}(\bm{x})$ of ${\cal Y}^n$ as
\begin{equation}
    {\cal T}_{\bm{V}}(\bm{x}) \coloneqq \{\bm{y}\in{\cal Y}^n:\bm{V}_{\bm{y}|\bm{x}}=\bm{V}\}. \label{eq:sequence_compatible_with_conditional_type}
\end{equation}
This is the set of strings that are compatible with the conditional type $\bm{V}$ given $\bm{x}$.  There is a distinguished element $\bm{V}^{\rm id}$ in ${\cal V}(\bm{x},{\cal X})$ called the identical conditional type that satisfies $\bm{V}^{\rm id}(y|x)=\delta_{xy}$ for any $(x,y)\in{\cal X}\times{\cal X}$.  In this case, the set $T_{\bm{V}^{\rm id}}$ is $T_{\bm{V}^{\rm id}}(\bm{x})=\{\bm{x}\}$.
From Ref.~\cite{Ahlswede1982}, the cardinality of the set ${\cal V}(\bm{x},{\cal X})$ is bounded from above by
\begin{equation}
    |{\cal V}(\bm{x},{\cal X})| \leq (n+1)^{|{\cal X}|^2}, \label{eq:cardinality_conditonal_type}
\end{equation}
which will be used later.

Let $\rho_{XB}$ be a classical-quantum (c-) state defined as
\begin{equation}
    \rho_{XB}\coloneqq \sum_{x\in{\cal X}}p(x)\ketbra{x}{x}_X\otimes \rho_B^{x}. \label{eq:c-q_state}
\end{equation}
Then, an i.i.d.~state $\rho_{XB}^{\otimes n}$ can be written as
\begin{equation}
    \rho_{XB}^{\otimes n} = \sum_{\bm{x}\in{\cal X}^n}p^n(\bm{x}) \ketbra{\bm{x}}{\bm{x}}_{X^n} \otimes \rho_{B^n}^{\bm{x}},
\end{equation}
where 
\begin{equation}
    p^n(\bm{x})\coloneqq \prod_{i=1}^n p(x_i),
\end{equation}
and
\begin{equation}
    \rho_{B^n}^{\bm{x}}\coloneqq \rho_{B}^{x_1}\otimes \cdots \otimes \rho_{B}^{x_n}. \label{eq:c-q_for_vec_x}
\end{equation}
Then, for each $\bm{x}$, we have 
\begin{equation}
    \rho_{B^n}^{\bm{x}_P} = (\rho_B^{1})^{\otimes m_1}\otimes\cdots\otimes(\rho_B^{|{\cal X}|})^{\otimes m_{|{\cal X}|}},
\end{equation}
where $m_i = n P_{\bm{x}}(i)$ for $i=1,\ldots,|{\cal X}|$.  By applying Eqs.~\eqref{eq:bound_universal_symmetric} and~\eqref{eq:quantum_coefficient_bound} recursively~\cite{Hayashi2009,Matsuura2024}, we have 
\begin{equation}
    \rho_{B^n}^{\bm{x}_P} \leq \prod_{i=1}^{|{\cal X}|} (m_i+1)^{\frac{(d+2)(d-1)}{2}} \sigma_{U,m_1}\otimes\cdots\otimes\sigma_{U,m_{|{\cal X}|}} \leq (n+1)^{\frac{(d+2)(d-1)}{2}|{\cal X}|}\sigma_{U,m_1}\otimes\cdots\otimes\sigma_{U,m_{|{\cal X}|}}.
\end{equation}
Furthermore, the right-hand side commutes with the left-hand side since the universal symmetric state $\sigma_{U,m_i}$ commutes with $O^{\otimes m_i}$ for any operator $O$.
We define $\sigma_{\bm{x}}$ for $\bm{x}\in{\cal X}^n$ as 
\begin{equation}
    \sigma_{\bm{x}} \coloneqq V_{s_{\bm{x}_P\to\bm{x}}}\sigma_{U,m_1}\otimes\cdots\otimes\sigma_{U,m_{|{\cal X}|}}V_{s_{\bm{x}_P\to\bm{x}}}^{\dagger}, \label{eq:conditional_universal}
\end{equation}
where $V_{s}$ is the unitary representation of $s\in S_n$ as mentioned earlier. By construction, $\sigma_{\bm{x}}$ is full rank.
From the relation $\bm{x} = s_{\bm{x}_P\to\bm{x}}(\bm{x}_P)$, the following holds for any $\bm{x}\in{\cal X}^n$:
\begin{equation}
    \rho_{B^n}^{\bm{x}} = V_{s_{\bm{x}_P\to\bm{x}}} \rho_{B^n}^{\bm{x}_P} V_{s_{\bm{x}_P\to\bm{x}}}^{\dagger} \leq (n+1)^{\frac{(d+2)(d-1)}{2}|{\cal X}|} \sigma_{\bm{x}}, \label{eq:sequence_dependent_bound}
\end{equation}
and $\rho_{B^n}^{\bm{x}}$ commutes with $\sigma_{\bm{x}}$.
For any type $P\in{\cal P}_n$, let $\sigma_{U,P}$ be defined as 
\begin{equation}
    \sigma_{U,P} \coloneqq \frac{1}{|{\cal T}_P|}\sum_{\bm{x}\in{\cal T}_P}  \sigma_{\bm{x}}. \label{eq:def_universal_type_symmetric}
\end{equation}
Then, the state $\sigma_{U,P}$ commutes with both $U^{\otimes n}$ for any $U\in\mathrm{SU}(d)$ and $V_{\sigma}$ for any $\sigma\in S_n$, and thus have the form $\sum_{\bm{n}\in Y_n^d}p_{\bm{n}}\sigma_{\bm{n}}$.  In particular, it commutes with $\sigma_{\bm{x}}$ for any $\bm{x}\in{\cal X}^n$.

\subsection{Entropic quantities} \label{sec:entropies}
Next, we introduce the information-theoretic quantities that characterize this task.  Let $D_{\alpha}(\rho\|\sigma)$ be the Petz $\alpha$-R\'enyi divergence defined as 
\begin{equation}
    D_{\alpha}(\rho\|\sigma) \coloneqq \begin{cases}\frac{1}{\alpha-1}\log\tr[\rho^{\alpha}\sigma^{1-\alpha}] & 0\leq \alpha < 1 \text{ or } \mathrm{supp}(\rho)\subseteq\mathrm{supp}(\sigma), \\
        \infty & \text{Otherwise}.
    \end{cases} \label{eq:alpha_Renyi_divergence}
\end{equation}
As $\alpha\rightarrow 1$, $D_{\alpha}(\rho\|\sigma)$ reduces to the quantum relative entropy $D(\rho\|\sigma)$ given by 
\begin{equation}
    D(\rho\|\sigma)\coloneqq\begin{cases} \tr[\rho\log\rho - \rho\log\sigma] & \mathrm{supp}(\rho)\subseteq\mathrm{supp}(\sigma), \\
        \infty & \text{Otherwise}.
    \end{cases}
\end{equation}
Similarly, the sandwiched $\alpha$-R\'enyi divergence $\widetilde{D}_{\alpha}(\rho\|\sigma)$ is defined as~\cite{Wilde2014,Martin2013}
\begin{equation}
    \widetilde{D}_{\alpha}(\rho\|\sigma) \coloneqq \begin{cases}
        \frac{\alpha}{\alpha-1}\log\|\sigma^{\frac{1-\alpha}{2\alpha}}\rho\sigma^{\frac{1-\alpha}{2\alpha}}\|_{\alpha} & 0\leq \alpha < 1 \text{ or } {\rm supp}(\rho)\subseteq{\rm supp}(\sigma), \\
        \infty & \text{Otherwise},
    \end{cases} \label{eq:sandwiched_renyi}
\end{equation}
where $\|T\|_{\alpha}$ denotes the Schatten $\alpha$-norm of an operator $T$.
The sandwiched $\alpha$-R\'enyi divergence also satisfies $\lim_{\alpha\to 1}\widetilde{D}_{\alpha}(\rho\|\sigma)=D(\rho\|\sigma)$.
It is shown in Ref.~\cite{Lin2015} that both $\alpha\mapsto D_{\alpha}(\rho\|\sigma)$ and $\alpha\mapsto \widetilde{D}_{\alpha}(\rho\|\sigma)$ are monotonically increasing.

The $\alpha$-mutual information $I_{\alpha}(A:B)$ is defined as 
\begin{equation}
    I_{\alpha}(A:B)_{\rho} =  
        \inf_{\sigma_B\in{\cal D}({\cal H}_B)} D_{\alpha}(\rho_{AB}\|\rho_A\otimes \sigma_B) 
     \label{eq:alpha_mutual}
\end{equation}
As $\alpha\rightarrow 1$, $I_{\alpha}(A:B)_{\rho}$ reduces to the conventional mutual information $I(A:B)_{\rho}$ defined as
\begin{equation}
I(A:B)_{\rho}\coloneqq D(\rho_{AB}\|\rho_A\otimes\rho_B).
\end{equation} 
For a c-q state $\rho_{XB}$ defined in Eq.~\eqref{eq:c-q_state}, we have the following explicit expression from the quantum Sibson's identity~\cite{Koenig2009, Sharma2013, Mosonyi2017, Cheng2022} for $\alpha\in[0,1)$: 
\begin{equation}
    I_{\alpha}(X:B)_{\rho} = \frac{\alpha}{\alpha-1}\log\mathrm{Tr}\!\left[\left(\sum_{x\in{\cal X}}p(x)\bigl(\rho_B^x\bigr)^{\alpha}\right)^{\frac{1}{\alpha}}\right] = D_{\alpha}(\rho_{XB}\|\rho_X\otimes\sigma_B^{*}),\label{eq:Sibson_identity} 
\end{equation}
where
\begin{equation}
    \sigma_B^{*} \coloneqq \frac{\left(\tr_X[\rho_X^{1-\alpha}\rho_{XB}^{\alpha}]\right)^{\frac{1}{\alpha}}}{\tr\left[\left(\tr_X[\rho_X^{1-\alpha}\rho_{XB}^{\alpha}]\right)^{\frac{1}{\alpha}}\right]}.
\end{equation}
Thus, for a c-q state, the infimum in Eq.~\eqref{eq:alpha_mutual} can be replaced with the minimum. The above expression implies that the function $\alpha\mapsto I_{\alpha}(X:B)_{\rho}$ is continuous, and Proposition~4 in Ref.~\cite{Cheng2022} shows that  $\alpha\mapsto I_{\alpha}(X:B)_{\rho}$ is non-decreasing.

For a c-q state $\rho_{XB}$, let $\breve{I}_{\alpha}(X:B)_{\rho}$ be sandwiched $\alpha$-Augustin information defined as~\cite{Beigi2023}
\begin{equation}
    \breve{I}_{\alpha}(X:B)_{\rho} \coloneqq \inf_{\sigma_B\in{\cal D}({\cal H}_B)} \sum_{x\in{\cal X}} p(x) \widetilde{D}_{\alpha}(\rho^x_B\|\sigma_B). \label{eq:def_augustin}
\end{equation}
Again, as $\alpha\rightarrow 1$, $\breve{I}_{\alpha}(X:B)_{\rho}$ reduces to the conventional mutual information $I(X:B)_{\rho}$~\cite{Cheng2022}. 
It is known that $\alpha\mapsto \breve{I}_{\alpha}(X:B)_{\rho}$ is continuous on $\alpha\in[1,2]$ (Lemma~12 of Ref.~\cite{Cheng2022error}) and monotonically increasing for $\alpha\geq 1/2$ (Proposition~5 of Ref.~\cite{Cheng2022}).

For a string $\bm{x}\in{\cal X}^n$, let $H(\bm{x})$ be the empirical entropy defined as 
\begin{equation}
    H(\bm{x}) \coloneqq -\sum_{y\in{\cal X}}P_{\bm{x}}(y)\log P_{\bm{x}}(y) = H(X)_{P_{\bm{x}}}. \label{eq:def_empirical_entropy}
\end{equation}
The empirical entropy thus depends only on the type of the sequence $\bm{x}$.
For any type $P\in{\cal P}_n$, the following is known to hold:
\begin{equation}
    (n+1)^{-|{\cal X}|}e^{n H(X)_P} \leq |{\cal T}_P|\leq e^{n H(X)_P}. \label{eq:type_cardinality_bound}
\end{equation}

\subsection{Further preliminaries for group representation} \label{sec:representation_theory}
Let $G$ be a finite group.  Consider two representations $(V_1,\pi_1)$ and $(V_2,\pi_2)$ of $G$. Then, an intertwiner $T:V_1\to V_2$ between two representations is defined as a linear map that satisfies $T \pi_1(g) v=\pi_2(g) T v$ for any $g\in G$ and $v\in V_1$. Note that $T=0$ trivially satisfies the above equation, which is called the trivial intertwiner. The following lemma states the existence/nonexistence of a nontrivial intertwiner.
\begin{lemma}[Schur's lemma] \label{lem:Schur}
    Given two irreducible unitary representations $(V_1,\pi_1)$ and $(V_2, \pi_2)$ of a finite group $G$, a nontrivial intertwiner $(0\neq)T:V_1\to V_2$ exists only when $(V_1,\pi_1)$ and $(V_2, \pi_2)$ are equivalent.  Furthermore, for two equivalent irreducible unitary representations, an intertwiner $T$ can always be taken to be unitary.
    If $(V_1',\pi_1')$ is a (not necessarily irreducible) unitary representation of $G$ that contains a subrepresentation equivalent to $(V_2,\pi_2)$, then there exists a surjective partial isometry $T':V_1'\to V_2$ such that $T'\pi_1'(g)=\pi_2(g)T'$ for any $g\in G$, i.e., $T'$ is an intertwiner.
\end{lemma}

Later, we will use the theory of induced representation. For a finite group $G$ and its subgroup $H$, let $(W,\pi)$ be a representation of $H$. Then, the representation $\mathrm{Ind}_{H}^G(\pi_W)$ of $G$ induced from the representation $(W,\pi)$ of $H$ is defined as follows~\cite{Fulton2004}. Let $\mathbb{C}[G]$ denote the group algebra of $G$, i.e., a complex algebra with each element of $G$ treated as a base. Then, regarding $\mathbb{C}[G]$ and $W$ as right- and left- $\mathbb{C}[H]$-modules, respectively, we define a module tensor product $\mathbb{C}[G]\otimes_{\mathbb{C}[H]} W$ as 
\begin{equation}
    \mathbb{C}[G]\otimes_{\mathbb{C}[H]} W\coloneqq\{g\otimes_{\mathbb{C}[H]} w, \forall g\in G, \forall w\in W: gh\otimes_{\mathbb{C}[H]} w = g\otimes_{\mathbb{C}[H]} \pi(h)w, \forall h \in H\}. \label{eq:induced_representation_space}
\end{equation}
Let $\rho:G\to \mathrm{GL}(\mathbb{C}[G]\otimes_{\mathbb{C}[H]} W) $ be a representation, where $\mathrm{GL}(V)$ denotes the set of linear operators on a vector space $V$, such that 
\begin{equation}
    \forall g, g'\in G, \forall w\in W,\qquad \rho(g) (g'\otimes_{\mathbb{C}[H]} w) = gg'\otimes_{\mathbb{C}[H]}w. \label{eq:left_multiplication}
\end{equation}
Then, the representation $(\mathbb{C}[G]\otimes_{\mathbb{C}[H]} W,\rho)$ of $G$ is called the induced representation and denoted by $\mathrm{Ind}_H^G(\pi_W)$. It is clear from the above that if $(W,\pi)$ is a unitary representation, $\mathrm{Ind}_H^G(\pi_W)$ is also a unitary representation. Note also that the representation space $\mathbb{C}[G]\otimes_{\mathbb{C}[H]} W$ is isomorphic to $\mathbb{C}[G/H]\otimes W$ as a vector space, where $G/H\coloneqq \{g H, g\in G\}$ denotes the set of left cosets of $H$ in $G$, while the representation of $G$ on $\mathbb{C}[G/H]\otimes W$ is not simple, unlike Eq.~\eqref{eq:left_multiplication}.
When $(W,\pi)$ is the regular representation of $H$, the induced representation $\mathrm{Ind}_H^G(\pi_W)$ is the regular representation of $G$. When $(W,\pi)$ is the trivial representation $(1,\mathbb{C})$ of $H$, then $\mathrm{Ind}_H^G(1_{\mathbb{C}})$ is called the permutation representation, which is the left multiplication of $G$ on $\mathbb{C}[G/H]$ (see Eq.~\eqref{eq:induced_representation_space}), acting as a permutation matrix. We later use $\mathrm{Ind}_H^G(1_{\mathbb{C}})$.
The following lemma directly follows from the Frobenius reciprocity theorem~\cite{Serre1977, Fulton2004}.
\begin{lemma} \label{lem:containing_irreducibles}
    Let $G$ be a finite group and $H$ be its subgroup. Let $(W,\pi)$ be an irreducible representation of $H$. Then, the induced representation $\mathrm{Ind}_H^G(\pi_W)$ contains every irreducible representation $(V,\Pi)$ of $G$ such that the restriction $\Pi\restriction_H:H\to\mathrm{GL}(V)$ of $\Pi$ to $H$ contains $(W,\pi)$ as a subrepresentation. The multiplicity of $(V,\Pi)$ in $\mathrm{Ind}_H^G(\pi_W)$ is equal to the multiplicity of $(W,\pi)$ in $(V,\Pi\restriction_H)$.
\end{lemma}

\subsection{Graph theory} \label{sec:graph_theory}
We introduce relevant graph-theoretic concepts.  We here consider an undirected multigraph $(V, E)$ with the vertex set $V$ and the edge set $E$, meaning that $E$ is a multiset (i.e., a set that allows multiple instances for each of its elements) of unordered pairs of vertices in $V$.  For an undirected multigraph $(V, E)$ with the vertex set $V$ and the edge set $E$, its adjacency matrix is a $|V|\times |V|$ matrix whose $(i,j)$-th entry is the number of edges that connect the vertices $v_i\in V$ and $v_j\in V$. (The adjacency matrix for an undirected multigraph is thus symmetric.)  The $(i,j)$-th entry is zero if $v_i$ and $v_j$ are disconnected.  A (multi)graph $(V, E)$ is said to be connected if any two vertices $u, v\in V$ have a connected path.  A (multi)graph is called $d$-regular if each vertex is attached to exactly $d$ edges.  

Consider a random walk on a connected undirected $d$-regular (multi)graph $(V, E)$, where a marker on a vertex moves to one of the neighboring vertices according to the uniform probability distribution $(1/d,\ldots,1/d)$ at each time step. The transition matrix of this random walk is given by the adjacency matrix divided by $d$, which has a unique maximum eigenvalue $1$ with the eigenvector corresponding to the fixed point of the transition matrix.  For an undirected $d$-regular graph $\Gamma$, let $1=\lambda_1\geq\lambda_2\geq\cdots\lambda_{|\Gamma|}\geq -1$ denote the spectrum of the transition matrix for $\Gamma$, and define $\lambda(\Gamma)$ as
\begin{equation}
    \lambda(\Gamma) \coloneqq \max\{\lambda_2, -\lambda_{|\Gamma|}\}. \label{eq:definition_lambda_graph}
\end{equation}
In particular, $\lambda(\Gamma)<1$ holds iff $\Gamma$ is connected and not bipartite.
An undirected $d$-regular (multi)graph $\Gamma$ with $n$ vertices and $\lambda(\Gamma)\leq \lambda$ is called $(n,d,\lambda)$-expander graph.  We are particularly interested in the smallest possible $\lambda$.  The optimal scaling of $\lambda$ for $d$ under $|V|\gg 1$ is known to be $\mathcal{O}(d^{-1/2})$ in the sense that there always exists a $d$-regular graph $\Gamma$ that has $\lambda(\Gamma) > \mathcal{O}(d^{-1/2})$ for a sufficiently large $|V|$.

Let $G$ be a finite group and $H$ be its subgroup. Then, $G$ has a transitive action on the set $G/H$ of left cosets of $H$ in $G$, which is in fact the permutation representation induced from the trivial representation $(\mathbb{C},1)$ of $H$. 
Let $S=\{S_i\}_{i\in I}$ be a multiset with $S_i\in G$ for any $i\in I$.  For simplicity, here we consider a symmetric subset, which means that for any $s\in S$, we have $s^{-1}\in S$.  Then, $\Gamma(G, G/H, S)$ denotes a Schreier graph, which is a multigraph with vertices $G/H$ and edges $(v, sv)$ for any $v \in G/H$ and $s\in S$ (counted with multiplicity).  Since $S$ is symmetric, $\Gamma(G,G/H,S)$ is an $|S|$-regular undirected graph.  (If $S$ is not symmetric, then $(v,sv)\in E$ does not imply $(sv,v)\in E$, and thus the edges have directions.) 
For this type of Schreier graph, we have the following.
\begin{lemma}\label{lem:transition_Schreier_graph}
    Suppose that $S\subseteq G$ is a symmetric subset of $G$ and that the Schreier graph $\Gamma(G, G/H, S)$ is connected. Then, the transition matrix $T$ of the Schreier graph $\Gamma(G, G/H, S)$ is given by 
    \begin{equation}
        T = \frac{1}{|S|}\sum_{g\in S}\mathrm{Ind}_H^G(1_{\mathbb{C}})(g), \label{eq:transition_matrix}
    \end{equation}
    which has a unique eigenspace with the eigenvalue one spanned by $\bm{1}\coloneqq(1,\ldots,1)$. Furthermore, the projector $P_{\bm{1}}$ onto this eigenspace is given by
    \begin{equation}
        P_{\bm{1}} = \frac{1}{|G|}\sum_{g\in G}\mathrm{Ind}_H^G(1_{\mathbb{C}})(g). \label{eq:fixed_point}
    \end{equation}
    Lastly, we have 
    \begin{equation}
        \lambda(\Gamma(G,G/H,S) = \sup_{\zeta\in\mathbb{C}[G/H]:\zeta^{\dagger}\zeta=1}\zeta^{\dagger}|T-P_{\bm{1}}|\zeta, \label{eq:second_largest_expression}
    \end{equation}
    where $\zeta^{\dagger}$ is a conjugate transponse of $\zeta\in\mathbb{C}[G/H]$.
\end{lemma}
\begin{proof}
    From the definition of the Schreier graph, the transition matrix is a uniform mixture of the permutation representation of $g\in S\subset G$ on $\mathbb{C}[G/H]$~\cite{Sabatini2022}. Thus, Eq.~\eqref{eq:transition_matrix} holds.
    Since $S$ is a symmetric subset, which implies that $\Gamma(G, G/H, S)$ is undirected, $T$ is a symmetric matrix. Furthermore, since $\Gamma(G, G/H, S)$ is connected, the second largest eigenvalue of $|T|$ is smaller than one (i.e., $\lambda(\Gamma(G,G/H,S))<1$). Since $\bm{1}$ is invariant under $\mathrm{Ind}_H^G(1_{\mathbb{C}})(g)$ for any $g\in G$, we have $T\bm{1}=\bm{1}$, meaning that the subspace spanned by $\bm{1}$ is the unique eigenspace of $T$ with the eigenvalue one. Since the right-hand side of Eq.~\eqref{eq:fixed_point} is the projection onto the trivial representation subspace of $G$ in ${\rm Ind}_{H}^{G}(1_{\mathbb{C}})$, and the trivial representation subspace has the multiplicity one from Lemma~\ref{lem:containing_irreducibles}, this subspace is spanned by $\bm{1}$ as well, which proves Eq.~\eqref{eq:fixed_point}. The last equation~\eqref{eq:second_largest_expression} directly follows from the definition of $\lambda(\Gamma)$ in Eq.~\eqref{eq:definition_lambda_graph}.
\end{proof}

The following lemma in Ref.~\cite{Sabatini2022} shows that a random construction of a Schreier graph ensures a good expansion property with a nonzero probability. 
\begin{lemma}[\cite{Sabatini2022}]\label{lem:random_construction}
    For any positive numbers $\delta$ and $\epsilon$, let $S=\{s_i\in G\}_{i\in{\cal I}}$ be a multiset of elements of $G$, where each element is picked up uniformly randomly from $G$ with $|{\cal I}| =  \left\lceil\frac{\log 4}{\epsilon^2}\log\left(\frac{2|G/H|}{\delta}\right)\right\rceil$.
    Then, we have
    \begin{equation}
        \mathrm{Pr}\bigl[\lambda\bigl(\Gamma(G, G/H, S\uplus S^{-1})\bigr) \geq \epsilon\bigr] \leq \delta,
    \end{equation}
    where $\uplus$ denotes the union of multisets, and $S^{-1}\coloneqq \{s^{-1}:s\in S\}$.
\end{lemma}

\section{Improved error exponent for universal classical-quantum channel coding} \label{sec:universal_c-q}
The universal coding for a c-q channel is first constructed in Ref.~\cite{Hayashi2009}.  A c-q channel $W$ is defined as a map ${\cal X}\ni x\mapsto W(x)\in{\cal D}({\cal H}_B)$. For any probability distribution $P$ over ${\cal X}$, we define a c-q state $W^P$ as 
\begin{equation}
    W^P \coloneqq \sum_{x\in{\cal X}}P(x)\ketbra{x}{x}_X\otimes W(x)
\end{equation}
Consider a channel code ${\cal C}_n$ for this c-q channel, where ${\cal C}_n$ is a triple $({\cal M}'_n,\phi_n,Y'_n)$ of a message set ${\cal M}'_n$, an encoder $\phi_n:{\cal M}'_n\to{\cal X}^n$, and a decoder $Y'_n:{\cal D}({\cal H}^{\otimes n})\to{\cal M}'_n$. Since the message set can be identified with the codebook $\phi_n({\cal M}'_n)\coloneqq \{\phi_n(m):m\in{\cal M}'_n\}$, we can redefine the code ${\cal C}_n$ as a pair $({\cal M}_n, Y_n)$ of a codebook ${\cal M}_n\subseteq{\cal X}^n$ and the decoding POVM $Y_n\coloneqq \{Y_n(\bm{x})\}_{\bm{x}\in{\cal M}_n}$. Hereafter, we use this notation for simplicity. The average decoding error $P_{\rm err}({\cal C}_n,W)$ for this code ${\cal C}_n$ over a channel $W$ is given by
\begin{equation}
    P_{\rm err}({\cal C}_n,W) \coloneqq \frac{1}{|{\cal M}_n|}\sum_{\bm{x}\in{\cal M}_n}\tr\left[W^{\otimes n}(\bm{x})(I-Y_n(\bm{x}))\right]
\end{equation}
Reference~\cite{Hayashi2009} showed that there exists a sequence of codes $\{{\cal C}_n\}_n$ with the asymptotic rate $R$ and the fixed type $P$, constructed without the knowledge of the c-q channel $W$, such that the decoding error $P_{\rm err}({\cal C}_n,W)$ satisfies
\begin{equation}
    \lim_{n\to \infty}-\frac{1}{n}\log P_{\rm err}({\cal C}_n,W) \geq \max_{\alpha\in[0,1]}\frac{\alpha}{1+\alpha}[I_{1-\alpha}(X:B)_{W^P} - R].
\end{equation}
Here, we prove a slightly improved error exponent as follows.

\begin{theorem} \label{theo:new_error_exponent}
    For any distribution $P$ over the input alphabets ${\cal X}$ and any real number $R$, there exists a sequence $\{{\cal C}_n\}_{n\in\mathbb{N}}$ of codes, where each ${\cal C}_n$ is a pair $({\cal M}_n,Y_n)$ of a codebook ${\cal M}_n$ and a decoder $Y_n$, such that
    \begin{equation}
        \lim_{n\to \infty}-\frac{1}{n}\log P_{\rm err}({\cal C}_n,W) \geq \max_{\alpha\in[0,1]} \alpha [I_{1-\alpha}(X:B)_{W^P}  - R], \label{eq:improved_error_exponent}
    \end{equation}
    and 
    \begin{equation}
        \lim_{n\to \infty} \frac{1}{n}\log|{\cal M}_n| = R, \label{eq:asymptotic_code_rate}
    \end{equation}
    for any c-q channel $W$. 
    The codes $\{{\cal C}_n\}_{n\in\mathbb{N}}$ do not depend on the channel $W$ and depend only on the probability distribution $P$, the rate $R$, and the dimension $d_B$ of the channel output.
    Thus, any rate $R < I(X:B)_{W^P}$ is universally achievable.
\end{theorem}

For the sequence of codes $\{{\cal C}_n\}$ that achieves the above error exponent, we use the same encoding strategy given in Ref.~\cite{Hayashi2009}, which requires a codebook ${\cal M}_n$ to satisfy a certain condition. To introduce the condition, we first define the following.
\begin{definition}[$(P, R, \delta)$-good set for $\bm{x}$] \label{def:good_set}
    For a type $P\in{\cal P}_n$, a positive real number $R<H(X)_P$, and a real number $\delta>0$, a set ${\cal L}_n$ of length-$n$ sequences is called $(P, R, \delta)$-good set for a length-$n$ sequence $\bm{x}$ if every element of ${\cal L}_n$ is in ${\cal T}_P$, $|{\cal L}_n|=\exp[nR-\delta]$, and
    \begin{equation}
        |{\cal T}_{\bm{V}}(\bm{x})\cap{\cal L}_n|\leq |{\cal T}_{\bm{V}}(\bm{x})|e^{-n(H(X)_P - R)},
    \end{equation}
    for any $\bm{V}\in{\cal V}(\bm{x},{\cal X})$, where ${\cal T}_{\bm{V}}(\bm{x})$ is defined in Eq.~\eqref{eq:sequence_compatible_with_conditional_type}.
\end{definition}
The above definition implies that ${\cal L}_n$ does not contain $\bm{x}$ to be $(P, R, \delta)$-good set for $\bm{x}$ since otherwise, the condition is violated for $\bm{V}=\mathrm{id}$.  
The codebook ${\cal M}_n$ in Ref.~\cite{Hayashi2009} is chosen so that for any codeword $\bm{x}\in{\cal M}_n$, ${\cal M}_n\setminus\{\bm{x}\}$ is $(P, R, \delta)$-good set for $\bm{x}$.
\begin{definition}[$(P, R, \delta)$-good codebook] \label{def:P-R_good}
    A set ${\cal M}_n$ of strings is said to be the $(P, R, \delta)$-good codebook if, for any $\bm{x}\in{\cal M}_n$, the set ${\cal M}_n\setminus\{\bm{x}\}$ is a $(P, R, \delta)$-good set for $\bm{x}$.
\end{definition} 
The existence of a sequence of codebooks that satisfy $(P, R, \delta(n))$-good codebook is shown in the literature~(Lemma~10.1 of \cite{Csiszar2015}, \cite{Hayashi2009, Hayashi2019}), and thus we state it without a proof.
\begin{lemma} \label{lem:code_book_choice}
    For a type $P$ and a positive real number $R$, there exists a sequence $\{{\cal M}_n\}_{n \geq N}$ of codebooks for a sufficiently large $N$ such that every ${\cal M}_n$ is a $(P, R, \delta_1(n))$-good codebook, where $\delta_1(n)=\Omega(\log n^{|{\cal X}|^2})$.
\end{lemma}
The condition of a $(P, R, \delta)$-good set is transformed into a more useful form by Hayashi~\cite{Hayashi2009}.
\begin{lemma} \label{lem:reduction_to_iid}
    Let ${\cal L}_n$ be a $(P, R, \delta)$-good set for $\bm{x}$.  Let $p_{{\cal L}_n}$ be a probability distribution defined as
    \begin{equation}
        p_{{\cal L}_n}(\bm{y})\coloneqq \begin{cases}
            |{\cal L}_n|^{-1} & \bm{y}\in {\cal L}_n, \\
            0 & \textrm{Otherwise}.
        \end{cases} \label{eq:p_M}
    \end{equation}
    Let $S_{\bm{x}}$ be a subgroup of $S_n$ that leaves $\bm{x}$ invariant.
    Then, for any $\bm{y}\in{\cal T}_P\setminus\{\bm{x}\}$, we have
    \begin{equation}
        \frac{1}{|S_{\bm{x}}|}\sum_{s\in S_{\bm{x}}} p_{{\cal L}_n}(s(\bm{y})) \leq  e^{-nH(X)_P + \delta}.
    \end{equation}
\end{lemma}
\begin{proof}
    Fix a string $\bm{y}\in{\cal T}_P\setminus\{\bm{x}\}$, and define a conditional type $\bm{V}\coloneqq\bm{V}_{\bm{y}|\bm{x}}$.
    From the fact that the joint type is invariant under the joint permutation and that $\bm{x}$ is invariant under any permutation $s\in S_{\bm{x}}$, for any $\bm{y}'\in{\cal T}_{\bm{V}}(\bm{x})$, there exists an element $s'\in S_{\bm{x}}$ such that $\bm{y}'=s'(\bm{y})$.  Thus, $\sum_{s\in S_{\bm{x}}}f(s(\bm{y})) = \sum_{s\in S_{\bm{x}}}f(s(\bm{y}'))$ for any function $f$ and any $\bm{y}'\in{\cal T}_{\bm{V}}(\bm{x})$. 
    Thus, from the definition of $p_{{\cal L}_n}$ in Eq.~\eqref{eq:p_M}, we have
    \begin{align}
        \frac{1}{|S_{\bm{x}}|}\sum_{s\in S_{\bm{x}}} p_{{\cal L}_n}(s(\bm{y})) &= \frac{1}{|{\cal T}_{\bm{V}}(\bm{x})|}\sum_{\bm{y}\in {\cal T}_{\bm{V}}(\bm{x})}\frac{1}{|S_{\bm{x}}|}\sum_{s\in S_{\bm{x}}} p_{{\cal L}_n}(s(\bm{y})) \\
        &= \frac{1}{|{\cal T}_{\bm{V}}(\bm{x})|}\sum_{\bm{y}\in {\cal T}_{\bm{V}}(\bm{x})} p_{{\cal L}_n}(\bm{y})  \label{eq:from_uniformity}\\
        &= \frac{|{\cal T}_{\bm V}(\bm{x})\cap{\cal L}_n|}{|T_{\bm{V}}(\bm{x})||{\cal L}_n|} \\
        &\leq |{\cal L}_n|^{-1}e^{-n(H(X)_P - R)} =e^{-n H(X)_P + \delta},
    \end{align}
    where we used in the second equality that the multiset $\{s(\bm{y}):s\in S_{\bm{x}},\bm{y}\in{\cal T}_{\bm{V}}(\bm{x})\}$ has precisely $|S_{\bm{x}}|$ duplicated elements for every $\bm{y}\in{\cal T}_{\bm{V}}(\bm{x})$.
\end{proof}
This lemma, combined with Lemma~\ref{lem:code_book_choice}, implies that the error probability of the channel coding with the encoding ${\cal M}_n$ may be bounded from above by that with the random coding up to an asymptotically negligible factor $e^{\delta_1(n)}$.

Next, we introduce a universal decoder that is an appropriate modification of the one used in Ref.~\cite{Matsuura2025} in the context of universal source compression with quantum side information. 
For a positive operator $A$ and a strictly positive operator $B$, we define an operator division $\frac{A}{B}$ as
\begin{equation}
    \frac{A}{B} \coloneqq \int_{0}^{\infty} d\lambda\, (B+\lambda I)^{-1} A (B+\lambda I)^{-1}. \label{eq:operator_division}
\end{equation}
The division $\frac{A}{B}$ has the following three properties~\cite{Beigi2023}:
\begin{enumerate}
    \item An operator $\frac{A}{B}$ is positive.
    \item It satisfies $\frac{A}{C}+\frac{B}{C}=\frac{A+B}{C}$.
    \item It satisfies $\frac{A}{A+B}\leq \frac{A}{B}$.
    \item Any unitary $U$ satisfies $U\frac{A}{B}U^{\dagger}=\frac{UAU^{\dagger}}{UBU^{\dagger}}$.
\end{enumerate}

With the above definition, we define a sequence of universal decoders $Y_n$ with POVM elements $Y_n(\bm{x})$ for any $\bm{x}\in{\cal M}_n$ given by
\begin{equation}
    Y_n(\bm{x}) \coloneqq \frac{\sigma_{\bm{x}}}{\sum_{\bm{y}\in{\cal M}_n}\sigma_{\bm{y}}}, \label{eq:decoding_POVM}
\end{equation}
where $\sigma_{\bm{x}}$ is defined in Eq.~\eqref{eq:conditional_universal}. Note that in Ref.~\cite{Hayashi2009}, a sequence of universal decoders $\tilde{Y}_n$ with
\begin{equation}
    \tilde{Y}_n(\bm{x})\coloneqq \left(\sum_{\bm{y}\in{\cal M}_n}\Pi(\bm{y})\right)^{-\frac{1}{2}} \Pi(\bm{x}) \left(\sum_{\bm{y}\in{\cal M}_n}\Pi(\bm{y})\right)^{-\frac{1}{2}} \label{eq:previous_universal_decoder}
\end{equation}
is used, where 
\begin{equation}
    \Pi(\bm{x})\coloneqq \{\sigma_{\bm{x}} - C_n \sigma_{U,n}\}_{+}, \label{eq:likelihood_ratio_test}
\end{equation}
for a constant $C_n$ to be adjusted, and $\{A\}_+$ for a Hermitian operator $A$ is a projection onto the eigenspaces with positive eigenvalues of $A$. We utilize this decoder later in the universal private channel coding.
With the encoder specified in Lemma~\ref{lem:code_book_choice} and the decoder specified in Eq.~\eqref{eq:decoding_POVM}, we achieve the error exponent in Theorem~\ref{theo:new_error_exponent} as follows.

\begin{proof}[Proof of Theorem~\ref{theo:new_error_exponent} ]
    We define the sequence $\{{\cal M}_n\}_{n\geq N}$ of codeboooks for a sufficiently large $N$ so that ${\cal M}_n$ is a $(P, R, \delta_1(n))$-good codebook, where the existence of such a sequence is ensured in Lemma~\ref{lem:codebook_required}. Then, Eq.~\eqref{eq:asymptotic_code_rate} is automatically satisfied.
    Furthermore, we use a POVM specified by Eq.~\eqref{eq:decoding_POVM} as a decoder.
    The average decoding failure probability $P_{\rm err}({\cal C}_n,W)$ can thus be given by  
    \begin{align}
       P_{\rm err}({\cal C}_n,W) &= \frac{1}{M_n}\sum_{\bm{x}\in{\cal M}_n} \tr\left[W^{\otimes n}(\bm{x}) \bigl(I - Y_n(\bm{x})\bigr)\right] \\
       &=  \frac{1}{M_n}\sum_{\bm{x}\in{\cal M}_n} \tr\left[W^{\otimes n}(\bm{x})\sum_{\bm{y}\in{\cal M}_n\setminus\{\bm{x}\}}Y_n(\bm{y})\right]. \label{eq:failure_explicit}
    \end{align}
    If Eq.~\eqref{eq:improved_error_exponent} holds, the last statement of the theorem directly follows from the fact that the function $\alpha\mapsto I_{\alpha}(X:B)_{\rho}$ is continuous and monotonically non-decreasing on $\alpha\in[0,1]$ and that $I_{\alpha}(X:B)_{\rho}\to I(X:B)_{\rho}$ as $\alpha\to 1$. 
    We prove Eq.~\eqref{eq:improved_error_exponent} with the following four steps.

    \noindent {\bf Step 1:}
    This step aims to show the following:
    \begin{equation}
        \sum_{\bm{y}\in{\cal M}_n\setminus\{\bm{x}\}}Y_n(\bm{y}) \leq M_n^{\alpha} \left(\frac{ \sum_{\bm{y}\in{\cal T}_P\setminus\{\bm{x}\}} p_{{\cal M}_n\setminus\{\bm{x}\}}(\bm{y})\, \sigma_{\bm{y}}}{\sigma_{\bm{x}}}\right)^{\alpha}. \label{eq:first_ineq_Theo1}
    \end{equation}
    First, for any POVM element $T$, we have $T\leq T^{\alpha}$ for any $\alpha\in[0,1]$.  Thus, we have
    \begin{align}
        \sum_{\bm{y}\in{\cal M}_n\setminus\{\bm{x}\}}Y_n(\bm{y}) \leq \left(\sum_{\bm{y}\in{\cal M}_n\setminus\{\bm{x}\}}Y_n(\bm{y})\right)^{\alpha} . \label{eq:property_fractional_power}
    \end{align}
    Notice that
    \begin{equation}
        \sum_{\bm{y}\in{\cal M}_n\setminus\{\bm{x}\}}Y_n(\bm{y}) = \frac{\sum_{\bm{y}\in{\cal M}_n\setminus\{\bm{x}\}} \sigma_{\bm{y}}}{\sum_{\bm{x}'\in{\cal M}_n}\sigma_{\bm{x}'}} \leq \frac{\sum_{\bm{y}\in{\cal M}_n\setminus\{\bm{x}\}} \sigma_{\bm{y}}}{\sigma_{\bm{x}}}, \label{eq:ineq_from_division}
    \end{equation}
    where we used the third property of the operator division $\frac{A}{B}$ defined in Eq.~\eqref{eq:operator_division}.
    From L\"owner-Heinz theorem, the function $t\mapsto t^{\alpha}$ is an operator monotone function for $\alpha\in[0,1]$. Combining this fact with Eqs.~\eqref{eq:property_fractional_power} and \eqref{eq:ineq_from_division}, we have
    \begin{align}
        \left(\sum_{\bm{y}\in{\cal M}_n\setminus\{\bm{x}\}}Y_n(\bm{y})\right)^{\alpha} &\leq \left(\frac{\sum_{\bm{y}\in{\cal M}_n\setminus\{\bm{x}\}} \sigma_{\bm{y}}}{\sigma_{\bm{x}}}\right)^{\alpha}  \\
        & = \left(\frac{(M_n-1) \sum_{\bm{y}\in{\cal T}_P\setminus\{\bm{x}\}} p_{{\cal M}_n\setminus\{\bm{x}\}}(\bm{y})\, \sigma_{\bm{y}}}{\sigma_{\bm{x}}}\right)^{\alpha} \\
        &\leq M_n^{\alpha}\left(\frac{\sum_{\bm{y}\in{\cal T}_P\setminus\{\bm{x}\}} p_{{\cal M}_n\setminus\{\bm{x}\}}(\bm{y})\, \sigma_{\bm{y}}}{\sigma_{\bm{x}}}\right)^{\alpha},
    \end{align}
    where the equality follows from the definition of $p_{{\cal L}}$ in Eq.~\eqref{eq:p_M}. Combining this with Eq.~\eqref{eq:property_fractional_power} leads to Eq.~\eqref{eq:first_ineq_Theo1}.

    \noindent {\bf Step 2:}
    This step aims to show
    \begin{equation}
        \frac{1}{|S_{\bm{x}}|}\sum_{s\in S_{\bm{x}}} V_s \left(\frac{\sum_{\bm{y}\in{\cal T}_P\setminus\{\bm{x}\}} p_{{\cal M}_n\setminus\{\bm{x}\}}(\bm{y})\, \sigma_{\bm{y}}}{\sigma_{\bm{x}}}\right)^{\alpha}V_s^{\dagger} \leq e^{\alpha \delta_1(n)} \sigma_{\bm{x}}^{-\alpha}\sigma_{U,P}^{\alpha}, \label{eq:convex_combination_bound}
    \end{equation}
    where $\sigma_{U,P}$ is defined in Eq.~\eqref{eq:def_universal_type_symmetric}.
    We first observe that
    \begin{align}
        &\frac{1}{|S_{\bm{x}}|}\sum_{s\in S_{\bm{x}}}V_s \left(\frac{\sum_{\bm{y}\in{\cal T}_P\setminus\{\bm{x}\}} p_{{\cal M}_n\setminus\{\bm{x}\}}(\bm{y})\, \sigma_{\bm{y}}}{\sigma_{\bm{x}}}\right)^{\alpha} V_s^{\dagger} \nonumber \\
        &= \frac{1}{|S_{\bm{x}}|}\sum_{s\in S_{\bm{x}}} \left(V_s\frac{ \sum_{\bm{y}\in{\cal T}_P\setminus\{\bm{x}\}} p_{{\cal M}_n\setminus\{\bm{x}\}}(\bm{y})\, \sigma_{\bm{y}}}{\sigma_{\bm{x}}}V_s^{\dagger}\right)^{\alpha} \\
        &= \frac{1}{|S_{\bm{x}}|}\sum_{s\in S_{\bm{x}}} \left(\frac{\sum_{\bm{y}\in{\cal T}_P\setminus\{\bm{x}\}} p_{{\cal M}_n\setminus\{\bm{x}\}}(\bm{y})\, V_s\sigma_{\bm{y}}V_s^{\dagger}}{V_s\sigma_{\bm{x}}V_s^{\dagger}}\right)^{\alpha} \\
        &= \frac{1}{|S_{\bm{x}}|}\sum_{s\in S_{\bm{x}}} \left(\frac{ \sum_{\bm{y}\in{\cal T}_P\setminus\{\bm{x}\}} p_{{\cal M}_n\setminus\{\bm{x}\}}(\bm{y})\, \sigma_{s(\bm{y})}}{\sigma_{\bm{x}}}\right)^{\alpha} \\
        &= \frac{1}{|S_{\bm{x}}|}\sum_{s'\in S_{\bm{x}}} \left(\frac{ \sum_{\bm{y}\in{\cal T}_P\setminus\{\bm{x}\}} p_{{\cal M}_n\setminus\{\bm{x}\}}(s'(\bm{y}))\, \sigma_{\bm{y}}}{\sigma_{\bm{x}}}\right)^{\alpha}, \label{eq:before_applying_lem6}
    \end{align}
    where the second equality follows from the fourth property of the operator division $\frac{A}{B}$, and the last equality follows from the fact that ${\cal T}_P\setminus\{\bm{x}\}$ is invariant under $S_{\bm{x}}$.
    From L\"owner-Heinz theorem, the function $t\mapsto t^{\alpha}$ is an operator concave function for $\alpha\in[0,1]$. Applying this to Eq.~\eqref{eq:before_applying_lem6}, we have
    \begin{equation}
        \frac{1}{|S_{\bm{x}}|}\sum_{s'\in S_{\bm{x}}} \left(\frac{ \sum_{\bm{y}\in{\cal T}_P\setminus\{\bm{x}\}} p_{{\cal M}_n\setminus\{\bm{x}\}}(s'(\bm{y}))\, \sigma_{\bm{y}}}{\sigma_{\bm{x}}}\right)^{\alpha} \leq  \left(\frac{\frac{1}{|S_{\bm{x}}|}\sum_{s'\in S_{\bm{x}}} \sum_{\bm{y}\in{\cal T}_P\setminus\{\bm{x}\}} p_{{\cal M}_n\setminus\{\bm{x}\}}(s'(\bm{y}))\, \sigma_{\bm{y}}}{\sigma_{\bm{x}}}\right)^{\alpha}. \label{eq:operator_concave_applied}
    \end{equation}
    Applying Lemma~\ref{lem:reduction_to_iid} to Eq.~\eqref{eq:operator_concave_applied}, we have
    \begin{align}
        \left(\frac{\frac{1}{|S_{\bm{x}}|}\sum_{s'\in S_{\bm{x}}} \sum_{\bm{y}\in{\cal T}_P\setminus\{\bm{x}\}} p_{{\cal M}_n\setminus\{\bm{x}\}}(s'(\bm{y}))\, \sigma_{\bm{y}}}{\sigma_{\bm{x}}}\right)^{\alpha} &\leq \left(\frac{\sum_{\bm{y}\in{\cal T}_P\setminus\{\bm{x}\}}e^{-n H(X)_P+\delta_1(n) }\sigma_{\bm{y}}}{\sigma_{\bm{x}}}\right)^{\alpha} \\
        &\leq \left(\frac{\sum_{\bm{y}\in{\cal T}_P}|{\cal T}_P|^{-1}e^{\delta_1(n) }\sigma_{\bm{y}}}{\sigma_{\bm{x}}}\right)^{\alpha} \\
        &= \left(\frac{e^{\delta_1(n)}\sigma_{U, P} }{\sigma_{\bm{x}}}\right)^{\alpha} \\
        &= e^{\alpha \delta_1(n)}\sigma_{\bm{x}}^{-\alpha}\sigma_{U,P}^{\alpha}, \label{eq:commutative_decomp}
    \end{align}
    where the second inequality follows from the operator monotonicity of $t\mapsto t^{\alpha}$, and the first equality follows from Eq.~\eqref{eq:def_universal_type_symmetric}, and the last equality follows from the fact that $\sigma_{U, P}$ and $\sigma_{\bm{x}}$ commutes for any $\bm{x}\in{\cal T}_P$. Combining Eqs.~\eqref{eq:before_applying_lem6}, \eqref{eq:operator_concave_applied}, and \eqref{eq:commutative_decomp}, we have Eq.~\eqref{eq:convex_combination_bound}.

    \noindent {\bf Step 3:}
    This step aims to show that 
    \begin{equation}
        \tr\left[W^{\otimes n}(\bm{x})\sum_{\bm{y}\in{\cal M}_n\setminus\bm{x}\}}Y_n(\bm{y})\right] \leq e^{\alpha\bigl(nR+\frac{|{\cal X}|(d+2)(d-1)}{2}\log(n+1)\bigr)}\tr\left[(W^{\otimes n}(\bm{x}))^{1 - \alpha}\sigma_{U,P}^{\alpha}\right]. \label{eq:before_transf_to_Renyi}
    \end{equation}
    From the cyclicity of the trace, we have
    \begin{align}
        \tr\left[W^{\otimes n}(\bm{x})\sum_{\bm{y}\in{\cal M}_n\setminus\{\bm{x}\}}Y_n(\bm{y})\right] &= \frac{1}{|S_{\bm{x}}|}\sum_{s\in S_{\bm{x}}}\tr\left[V_s W^{\otimes n}(\bm{x})V_s^{\dagger}V_s\sum_{\bm{y}\in{\cal M}_n\setminus\{\bm{x}\}}Y_n(\bm{y})V_s^{\dagger}\right] \\
        &= \tr\left[W^{\otimes n}(\bm{x})\frac{1}{|S_{\bm{x}}|}\sum_{s\in S_{\bm{x}}}V_s\left(\sum_{\bm{y}\in{\cal M}_n\setminus\{\bm{x}\}}Y_n(\bm{y})\right)V_s^{\dagger}\right]. \label{eq:before_applying_steps_12}
    \end{align}
    Then, applying Eqs.~\eqref{eq:first_ineq_Theo1} and \eqref{eq:convex_combination_bound} to Eq.~\eqref{eq:before_applying_steps_12}, we have 
    \begin{equation}
        \tr\left[W^{\otimes n}(\bm{x})\frac{1}{|S_{\bm{x}}|}\sum_{s\in S_{\bm{x}}}V_s\left(\sum_{\bm{y}\in{\cal M}_n\setminus\{\bm{x}\}}Y_n(\bm{y})\right)V_s^{\dagger}\right] \leq M_n^{\alpha} e^{\alpha(\delta_1(n)} \tr[W^{\otimes n}(\bm{x})\sigma_{\bm{x}}^{-\alpha}\sigma_{U,P}^{\alpha}]. \label{eq:after_random_permutation}
    \end{equation}
    From Eq.~\eqref{eq:sequence_dependent_bound}, we have
    \begin{equation}
        W^{\otimes n}(\bm{x})\sigma_{\bm{x}}^{-\alpha} \leq (n+1)^{\alpha\frac{|{\cal X}|(d+2)(d-1)}{2}} \left(W^{\otimes n}(\bm{x})\right)^{1-\alpha}. \label{eq:1-alpha_bound}
    \end{equation}
    Combining Eqs.~\eqref{eq:before_applying_steps_12}, \eqref{eq:after_random_permutation}, and \eqref{eq:1-alpha_bound} with $M_n\leq \exp[nR-\delta_1(n)]$ from the definition of the $(P, R, \delta_1(n))$-good codebook in Def.~\eqref{def:P-R_good}, we have Eq.~\eqref{eq:before_transf_to_Renyi}.

    \noindent {\bf Step 4:}
    We finally prove the following, which immediately implies Eq.~\eqref{eq:improved_error_exponent}: for any $\alpha\in[0,1]$,
    \begin{equation}
        P_{\rm err}({\cal C}_n,W) \leq \exp\left[-n\left(\alpha[I_{1-\alpha}(X:B)_{W^P}-R] - \frac{\alpha (d+2)(d-1)+2}{2}\frac{|{\cal X}|\log(n+1)}{n} \right)\right]. \label{eq:Theorem_statement_precise}
    \end{equation}
    Since $\sigma_{U,P}$ is permutation-invariant, we have 
    \begin{align}
        \tr\left[(W^{\otimes n}(\bm{x}))^{1 - \alpha}\sigma_{U,P}^{\alpha}\right] &=\frac{1}{|S_n|}\sum_{s\in S_n}\tr\left[V_s(W^{\otimes n}(\bm{x}))^{1 - \alpha}V_s^{\dagger}\sigma_{U,P}^{\alpha}\right] \\
        &= \frac{1}{|S_n|}\sum_{s\in S_n}\tr\left[\bigl(W^{\otimes n}(s(\bm{x}))\bigr)^{1 - \alpha}\sigma_{U,P}^{\alpha}\right] \\
        &= \frac{1}{|{\cal T}_P|} \sum_{\bm{x}\in{\cal T}_P} \tr\left[\bigl(W^{\otimes n}(\bm{x})\bigr)^{1 - \alpha}\sigma_{U,P}^{\alpha}\right] \\
        &\leq e^{-nH(X)_P + |{\cal X}|\log(n+1)} \sum_{\bm{x}\in{\cal T}_P} \tr\left[\bigl(W^{\otimes n}(\bm{x})\bigr)^{1 - \alpha}\sigma_{U,P}^{\alpha}\right], \label{eq:after_permutation}
    \end{align}
    where the last inequality follows from Eq.~\eqref{eq:type_cardinality_bound}. Using $P^n(\bm{x})=e^{-nH(X)_P}$ for any $\bm{x}\in{\cal T}_P$, we have
    \begin{align}
        e^{-nH(X)_P} \sum_{\bm{x}\in{\cal T}_P} \tr\left[\bigl(W^{\otimes n}(\bm{x})\bigr)^{1 - \alpha}\sigma_{U,P}^{\alpha}\right] &= \sum_{\bm{x}\in{\cal T}_P} P^n(\bm{x})\,\tr\left[\bigl(W^{\otimes n}(\bm{x})\bigr)^{1 - \alpha}\sigma_{U,P}^{\alpha}\right] \\
        &\leq  \tr\left[\left(\sum_{x\in{\cal X}}P(x) \bigl(W(x)\bigr)^{1-\alpha}\right)^{\otimes n}\sigma_{U,P}^{\alpha}\right] \label{eq:type_to_iid}
    \end{align}
    From Lemma~2 of Ref.~\cite{Hayashi2009}, we have
    \begin{align}
        \tr\left[\left(\sum_{x\in{\cal X}}P(x) \bigl(W(x)\bigr)^{1-\alpha}\right)^{\otimes n}\sigma_{U,P}^{\alpha}\right] &\leq \left(\tr\left[\sum_{x\in{\cal X}}P(x) \left(W(x)\right)^{1 - \alpha}\right]^{\frac{1}{1-\alpha}}\right)^{n(1-\alpha)} \\
        &= \exp\left[-n\alpha I_{1-\alpha}(X:B)_{W^P}\right], \label{eq:bound_by_alpha_mutual}
    \end{align}
    where the equality follows from Eq.~\eqref{eq:Sibson_identity}.
    Combining Eqs.~\eqref{eq:failure_explicit}, \eqref{eq:before_transf_to_Renyi}, \eqref{eq:after_permutation}, \eqref{eq:type_to_iid}, and \eqref{eq:bound_by_alpha_mutual}, we prove Eq.~\eqref{eq:Theorem_statement_precise}.
\end{proof}

\section{Universal classical-quantum resolvability coding} \label{sec:universal_resolvability}
In this section, we describe the universal version of the classical-quantum (c-q) resolvability coding.  The related concept is the soft-covering lemma, proved against the random construction of codewords~\cite{Ahlswede2002,Datta2010,Cheng2022error,Sen2025}.

What we will show in this section is the following theorem.
\begin{theorem}[Universal resolvability] \label{theo:universal_resolvability}
    For any distribution $P$ over the input alphabets ${\cal X}$ and any real number $R$, the sequence $\{{\cal M}_n\}_{n\in\mathbb{N}}$ of codebooks that is $(P, R, \delta_2(n))$-radical spectral expander (defined below in Def.~\ref{def:radical_spectral_expander}) satisfies
    \begin{equation}
        -\lim_{n\to\infty}\frac{1}{n}\log\left\|\frac{1}{|{\cal M}_n|}\sum_{\bm{x}\in{\cal M}_n}W^{\otimes n}(\bm{x}) - \overline{W}^n_P \right\|_1 \geq \max_{\alpha\in[1,2]} \frac{\alpha-1}{\alpha}(R - \breve{I}_{\alpha}(X:B)_{W^P}), \label{eq:exponent_resolvability}
    \end{equation}
    where $\overline{W}^n_P\coloneqq \frac{1}{|{\cal T}_P|}\sum_{\bm{x}\in{\cal T}_P}W^{\otimes n}(\bm{x})$, and
    \begin{equation}
        \lim_{n\to\infty}\frac{1}{n}\log |{\cal M}_n| = R. \label{eq:asymptotic_rate_resolvability}
    \end{equation}
    The choice of the $\{{\cal M}_n\}_{n\in\mathbb{N}}$ does not depend on the channel $W$ and depends only on the probability distribution $P$ and the rate $R$. Thus, any rate $R>I(X:B)_{W^P}$ is universally achievable.
\end{theorem}

Unlike the result in Ref.~\cite{Sen2025}, our theorem states the existence of a sequence of fixed codebooks that satisfy Eq.~\eqref{eq:exponent_resolvability}.
Notice also that the above theorem is independent of the dimension of the channel output, unlike Theorem~\ref{theo:new_error_exponent}. This implies that the universal resolvability in Theorem~\ref{theo:universal_resolvability} holds even under an infinite-dimensional channel output.
We first define the following relevant property for a sequence of codebooks.
\begin{definition}[$(P, R, \delta)$-radical spectral expander] \label{def:radical_spectral_expander}
    For a type $P\in{\cal P}_n$, a positive real number $R < H(X)_P$, and a real number $\delta>0$, a codebook ${\cal M}_n$ is called $(P, R, \delta)$-radical spectral expander if there exists a multiset $S'$ of elements in $S_n$ that satisfies the following conditions.
    \begin{enumerate}
        \item ${\cal M}_n=S'(\bm{x}_P)$, where $S'(\bm{x})\coloneqq\{s(\bm{x})\}_{s\in S'}$.
        \item $|{\cal M}_n|=|S'|=\exp[nR + \delta]$.
        \item A multiset $S'$ is a symmetric subset of $S_n$ in the sense defined in Sec.~\ref{sec:graph_theory}.
        \item The Schreier graph $\Gamma(S_n, S_n/S_{\bm{x}_P}, S')$ with a set $S_n/S_{\bm{x}_P}$ of left cosets of $S_{\bm{x}_P}$ in $S_n$, where $S_{\bm{x}_P}$ is as defined in Eq.~\eqref{eq:isotropy}, satisfies
        \begin{equation}
            \lambda\bigl(\Gamma(S_n, S_n/S_{\bm{x}_P}, S')\bigr) \leq \exp\!\left(-\frac{nR}{2}\right). \label{eq:radical_condition}
        \end{equation}
    \end{enumerate}
\end{definition}

We will state that by randomly picking up codewords from the set ${\cal T}_P$, the resulting codebook satisfies the $(P, R, \delta)$-radical spectral expander with a nonzero probability.
In fact, as a consequence of Lemma~\ref {lem:random_construction}, we have the following.
\begin{lemma} \label{lem:existence}
    For a type $P$ and a positive real number $R<H(X)_P$, there exists a sequence $\{{\cal M}_n\}_{n\geq N}$ of codebooks for a sufficiently large $N$ such that every ${\cal M}_n$ is a $(P, R, \delta_2(n))$-radical spectral expander with $\delta_2(n)=\Omega(\log n)$.
\end{lemma}
To prove this lemma, we first prove the following.
\begin{lemma}\label{lem:bijection_coset_type}
    There exists a bijection from the set $S_n/S_{\bm{x}_P}$ of left cosets of $S_{\bm{x}_P}$ in $S_n$ to the set ${\cal T}_P$ of strings with the type $P$.
\end{lemma}
\begin{proof}
    For any left coset $S\in S_n/S_{\bm{x}_P}$, we take a representative $\sigma_S\in S$ and define a map $S_n/S_{\bm{x}_P}\ni S\mapsto \sigma_S(\bm{x}_P) \in {\cal T}_P$. This map is well-defined since, for any $s_1,s_2\in S$, we have $s_1(\bm{x}_P)=s_2(\bm{x}_P)$, which implies that the map defined above is independent of the choice of the representative $\sigma_S\in S$. Furthermore, for any $\bm{x}\in{\cal T}_P$, there exists an element $s_{\bm{x}_P\to\bm{x}}\in S_n$ such that $\bm{x}=s_{\bm{x}_P\to\bm{x}}(\bm{x}_P)$, which implies that the map defined above is surjective. Finally, for any $S_1,S_2\in S_n/S_{\bm{x}_P}$ satisfying $S_1\neq S_2$, we have $\sigma_{S_1}(\bm{x}_P)\neq\sigma_{S_2}(\bm{x}_P)$, which implies that the map defined above is injective.
\end{proof}
With this lemma, we prove Lemma~\ref{lem:existence}.
\begin{proof}[Proof of Lemma~\ref{lem:existence}]
    From Lemma~\ref{lem:random_construction}, a random construction ensures
    \begin{equation}
        \mathrm{Pr}\left[\lambda\bigl(\Gamma(S_n, S_n/S_{\bm{x}_P}, S')\bigr) \leq \epsilon\right] > 0
    \end{equation}
    with $|S'|=\frac{\Omega(\log|S_n/S_{\bm{x}_P}|)}{\epsilon^2}$. From Lemma~\ref{lem:bijection_coset_type}, we have $|S_n/S_{\bm{x}_P}|=|{\cal T}_P|\leq e^{nH(X)_P}$. Thus, by setting $\epsilon=\exp(-nR/2)$ and repeating the random construction until success, we obtain a multiset $S'$ that satisfies the second to the fourth conditions in Def.~\ref{def:radical_spectral_expander}. By defining ${\cal M}_n=S'(\bm{x}_P)$ for such $S'$, we prove the statement.
\end{proof}

In the following, we show that a sequence of codebooks that satisfy the $(P, R, \delta_2(n))$-radical spectral expander property for a sufficiently large $n$ achieves the universal channel resolvability.
For an operator $A$ on ${\cal H}_B^{\otimes n}$, let ${\cal V}_s(A)\coloneqq V_s A V_s^{\dagger}$ denotes a unitary channel.
For any subgroup $S$ of $S_n$, we define $\mathbb{E}_{S}(A)$ as
\begin{equation}
    \mathbb{E}_{S}(A) = \frac{1}{|S|}\sum_{s\in S}{\cal V}_s(A). \label{eq:subgroup_averaging}
\end{equation}
Furthermore, for any subset $S'$ of $S_n$, we define the linear map $\Theta_{S'}:{\cal S}_p({\cal H}_B^{\otimes n})\to{\cal S}_p({\cal H}_B^{\otimes n})$ for the Schatten-class operators ${\cal S}_p({\cal H}_B^{\otimes n})$ as
\begin{equation}
    \Theta_{S'}(A) \coloneqq \frac{1}{|{\cal M}|}\sum_{s\in S'} {\cal V}_{s} \circ \mathbb{E}_{ S_{\bm{x}_P}}(A) - \mathbb{E}_{S_n}(A). 
    \label{eq:def_theta_map}
\end{equation}
Then, we have the following.
\begin{lemma} \label{lem:two_norm_bound}
    Let ${\cal M}_n$ be a codebook that is $(P, R, \delta)$-radical spectral expander, and let $S'$ be a multiset of elements in $S_n$ that satisfies all the conditions in Def.~\ref{def:radical_spectral_expander} for ${\cal M}_n$.  Then, the following inequality holds for any $\eta\in {\cal S}_2({\cal H}_B^{\otimes n})$:
    \begin{equation}
        \|\Theta_{S'}(\eta)\|_2 \leq e^{-\frac{nR}{2}}\|\eta\|_2. \label{eq:two_norm_bound}
    \end{equation}
\end{lemma}
\begin{proof}
    We will prove the lemma with the following three steps. In Step~1, a superoperator representation of the group $S_n$ and its relevant subrepresentation are considered. Equation~\eqref{eq:two_norm_bound} is then reduced to the evaluation of the operator norm of the superoperator representation of Eq.~\eqref{eq:def_theta_map}.  In Step~2, the induced representation from the irreducible trivial representation of the subgroup $S_{\bm{x}_P}$ of $S_n$ is considered, and the spectrum of a transition matrix that corresponds to $S'$ is analyzed with a Schreier graph introduced in Sec.~\ref{sec:graph_theory}. In Step~3, evaluation of the operator norm of the superoperator in Step~1 is reduced to the spectrum of the transition matrix studied in Step~2, which then implies Eq.~\eqref{eq:two_norm_bound}
    
    \noindent {\bf Step 1:} The aim of this step is the preparation of the proof. Let us regard ${\cal S}_2({\cal H}_B^{\otimes n})$ as a Hilbert space with a Hilbert-Schmidt inner product $\braket{A|B}=\mathrm{Tr}[A^{\dagger}B]$ for $A, B\in {\cal S}_2({\cal H}_B^{\otimes n})$.  Then, the unitary channel ${\cal V}_s(\cdot)$ for any $s\in S_n$ can be regarded as a unitary operator ${\cal V}_s:{\cal S}_2({\cal H}_B^{\otimes n})\to{\cal S}_2({\cal H}_B^{\otimes n})$ on this Hilbert space with its Hermitian adjoint ${\cal V}_s^{\dagger}$ given by ${\cal V}^{\dagger}_s(A)\coloneqq V_s^{\dagger}A V_s={\cal V}_{s^{-1}}(A)$ from the defining relation
    \begin{equation}
        \braket{A|{\cal V}_s(B)}=\mathrm{Tr}[A^{\dagger}\,{\cal V}_s(B)]=\mathrm{Tr}[{\cal V}_s^{\dagger}(A^{\dagger})\, B]=\braket{{\cal V}_s^{\dagger}(A)|B}.
    \end{equation}
    Furthermore, for any subgroup $S$ of $S_n$, the channel $\mathbb{E}_{S}(\cdot)$ is represented as a projection operator $\mathbb{E}_{S}:{\cal S}_2({\cal H}_B^{\otimes n})\to{\cal S}_2({\cal H}_B^{\otimes n})$ on this Hilbert space since 
    \begin{align}
        \braket{A|\mathbb{E}_{S}(B)}&=\frac{1}{|S|}\sum_{s\in S} \tr[A^{\dagger} {\cal V}_s(B)] \\
        &= \frac{1}{|S|}\sum_{s\in S} \tr[{\cal V}_{s^{-1}}(A^{\dagger}) B] \\
        &= \frac{1}{|S|}\sum_{s'\in S} \tr[{\cal V}_{s'}(A^{\dagger}) B]\\
        &= \braket{\mathbb{E}_{S}(A)|B},
    \end{align}
    which implies that $\mathbb{E}_{S}$ is self-adjoint, and ${\cal V}_s\circ \mathbb{E}_{S}=\mathbb{E}_{sS}=\mathbb{E}_{S}$ for any $s\in S$, which implies $\mathbb{E}_{S}\circ \mathbb{E}_{S}=\mathbb{E}_{S}$.  Let $\widetilde{\cal M}_n$ be a linear operator on the Hilbert space ${\cal S}_2({\cal H}_B^{\otimes n})$ defined as
    \begin{equation}
    \widetilde{\cal M}_n(\eta) \coloneqq \frac{1}{|S'|}\sum_{s\in S'}{\cal V}_{s}(\eta)=\frac{1}{|{\cal M}_n|}\sum_{s\in S'}{\cal V}_{s}(\eta), \label{eq:def_tilde_M_n}
    \end{equation}
    which is a convex mixture of unitary operators $\{{\cal V}_s:s\in S'\}$ on ${\cal S}_2({\cal H}_B^{\otimes n})$.  Then, what we need to prove (i.e., Eq.~\eqref{eq:two_norm_bound}) is equivalent to proving
    \begin{equation}
        \|\widetilde{\cal M}_n\circ \mathbb{E}_{S_{\bm{x}_P}} - \mathbb{E}_{S_n}\| \leq e^{-\frac{nR}{2}}, \label{eq:the_goal}
    \end{equation}
    where $\|\cdot\|$ denotes the operator norm for the set of linear operators on the Hilbert space ${\cal S}_2({\cal H}_B^{\otimes n})$.
    
    \noindent {\bf Step 2:} Let ${\cal K}\coloneqq \mathbb{E}_{S_{\bm{x}_P}}({\cal S}_2({\cal H}_B^{\otimes n}))$ be a Hilbert subspace of ${\cal S}_2({\cal H}_B^{\otimes n})$, and let ${\cal J}$ be defined as
    \begin{equation}
        {\cal J} \coloneqq \left\{\sum_{s\in S_n}\bm{v}_s:\bm{v}_s\in{\cal V}_s({\cal K}), \forall s\in S_n \right\}, \label{eq:def_of_J}
    \end{equation}
    where ${\cal V}_s({\cal K})=\{{\cal V}_s(\eta):\eta\in{\cal K}\}$ for each $s\in S_n$ is a subspace of ${\cal S}_2({\cal H}_B^{\otimes n})$. Then, the aim of this section is to show 
    \begin{equation}
        \|\widetilde{\cal M}_n\circ \mathbb{E}_{S_{\bm{x}_P}} - \mathbb{E}_{S_n}\|^2 \leq \sup_{i\in\{1,\ldots,l\}}\sup_{\eta\in J_i:\braket{\eta|\eta}\leq1} \Braket{(\widetilde{\cal M}_n - \mathbb{E}_{S_n})\restriction_{J_i}(\eta)|(\widetilde{\cal M}_n - \mathbb{E}_{S_n})\restriction_{J_i}(\eta)},\label{eq:modified_goal}
    \end{equation}
    where $\{J_1,\ldots,J_l:\forall i \in \{1,\ldots,l\},J_i\in{\cal J}\}$ is the set of irreducible unitary representation subspaces for the representation $\{{\cal V}_s:s\in S_n\}$ of $S_n$ on ${\cal J}$, and $\restriction_{J_i}$ denotes the restriction of a linear operator to the subspace $J_i$.
    First, from $\mathbb{E}_{S_{\bm{x}_P}}\circ \mathbb{E}_{S_n}=\mathbb{E}_{S_n} = \mathbb{E}_{S_n}\circ\mathbb{E}_{S_{\bm{x}_P}}$, we have
    \begin{align}
        \|\widetilde{\cal M}_n\circ \mathbb{E}_{S_{\bm{x}_P}} - \mathbb{E}_{S_n}\|^2 &= \|(\widetilde{\cal M}_n- \mathbb{E}_{S_n})\circ \mathbb{E}_{S_{\bm{x}_P}} \|^2 \\
        &= \sup_{\tau\in {\cal S}_2({\cal H}_B^{\otimes n}):\braket{\tau|\tau}\leq 1}\Braket{(\widetilde{\cal M}_n- \mathbb{E}_{S_n})\circ \mathbb{E}_{S_{\bm{x}_P}}(\tau)|(\widetilde{\cal M}_n- \mathbb{E}_{S_n})\circ \mathbb{E}_{S_{\bm{x}_P}}(\tau)} \\
        &= \sup_{\eta\in {\cal K}:\braket{\eta|\eta}\leq 1} \Braket{(\widetilde{\cal M}_n - \mathbb{E}_{S_n})(\eta)|(\widetilde{\cal M}_n - \mathbb{E}_{S_n})(\eta)},  \label{eq:maximization_over_K}
    \end{align}
    where the second equality follows from the definition of the operator norm.
    Note that the unitary representation $\{{\cal V}_s:s\in S_{\bm{x}_P}\}$ of the subgroup $S_{\bm{x}_P}\subseteq S_n$ acts trivially on ${\cal K}$. 
    The subspace ${\cal J}$ of ${\cal S}_2({\cal H}_B^{\otimes n})$ defined in Eq.~\eqref{eq:def_of_J} is the smallest subspace closed under the action of $\{{\cal V}_s:s\in S_n\}$ that contains ${\cal K}$ since ${\cal V}_{s}\circ{\cal V}_{s'}({\cal K})={\cal V}_{ss'}({\cal K})\in{\cal J}$ for any $s, s'\in S_n$.
    Let us consider the subrepresentation $\{{\cal V}_s\restriction_{\cal J}:s\in S_n\}$ of $\{{\cal V}_s:s\in S_n\}$ on the subspace ${\cal J}$. Then, since $\widetilde{M}_n$ and $\mathbb{E}_{S_n}$ are convex combinations of $\{{\cal V}_s:s\in S_n\}$ as in Eqs.~\eqref{eq:subgroup_averaging} and \eqref{eq:def_tilde_M_n}, we have 
    \begin{equation}
        \forall\eta\in{\cal K},\qquad (\widetilde{\cal M}_n - \mathbb{E}_{S_n})(\eta) = (\widetilde{\cal M}_n - \mathbb{E}_{S_n})\restriction_{\cal J}(\eta).
    \end{equation}
    Applying this to Eq.~\eqref{eq:maximization_over_K}, we have
    \begin{align}
        &\sup_{\eta\in {\cal K}:\braket{\eta|\eta}\leq 1} \Braket{(\widetilde{\cal M}_n - \mathbb{E}_{S_n})(\eta)|(\widetilde{\cal M}_n - \mathbb{E}_{S_n})(\eta)}\nonumber \\ 
        &= \sup_{\eta\in {\cal K}:\braket{\eta|\eta}\leq 1} \Braket{(\widetilde{\cal M}_n - \mathbb{E}_{S_n})\restriction_{\cal J}(\eta)|(\widetilde{\cal M}_n - \mathbb{E}_{S_n})\restriction_{\cal J}(\eta)}\\
        &\leq \sup_{\eta\in {\cal J}:\braket{\eta|\eta}\leq 1} \Braket{(\widetilde{\cal M}_n - \mathbb{E}_{S_n})\restriction_{\cal J}(\eta)|(\widetilde{\cal M}_n - \mathbb{E}_{S_n})\restriction_{\cal J}(\eta)}. \label{eq:rep_stab}
    \end{align}
    We further consider the irreducible decomposition of the subspace ${\cal J}$. From Maschke's theorem (Theorem~2 of Ref.~\cite{Serre1977}, Corollary~1.6 of Ref.~\cite{Fulton2004}), each $J_i$ ($i\in\{1,\ldots,l\}$) is not mixed by $\{{\cal V}_s:s\in S_n\}$, i.e., the unitary representation $\{{\cal V}_s:s\in S_n\}$ is completely reducible. Again, since $\widetilde{\cal M}_n$ and $\mathbb{E}_{S_n}$ are convex combinations of $\{{\cal V}_s:s\in S_n\}$ as in Eqs.~\eqref{eq:subgroup_averaging} and \eqref{eq:def_tilde_M_n}, we have 
    \begin{equation}
    \begin{split}
        &\sup_{\eta\in {\cal J}:\braket{\eta|\eta}\leq 1} \Braket{(\widetilde{\cal M}_n - \mathbb{E}_{S_n})\restriction_{\cal J}(\eta)|(\widetilde{\cal M}_n - \mathbb{E}_{S_n})\restriction_{\cal J}(\eta)} \\
        &= \sup_{i\in\{1,\ldots,l\}}\sup_{\eta\in J_i:\braket{\eta|\eta}\leq1} \Braket{(\widetilde{\cal M}_n - \mathbb{E}_{S_n})\restriction_{J_i}(\eta)|(\widetilde{\cal M}_n - \mathbb{E}_{S_n})\restriction_{J_i}(\eta)}.
        \end{split}
    \end{equation}
    Combining this with Eqs.~\eqref{eq:maximization_over_K}, \eqref{eq:rep_stab}, we have Eq.~\eqref{eq:modified_goal}.

    \noindent{\bf Step 3:}
    This step aims to show
    \begin{equation}
        \sup_{\zeta\in\mathbb{C}[S_n/S_{\bm{x}_P}]:\,\zeta^{\dagger}\zeta \leq 1} \zeta^{\dagger} (T-P_{\bm{1}})^2 \zeta \leq e^{-nR}, \label{eq:end_step_3}
    \end{equation}
    where $T$ is the transition matrix of the Schreier graph $\Gamma(S_n,S_n/S_{\bm{x}_P},S')$ and $P_{\bm{1}}$ denotes the projection to the eigenspace of $T$ with the eigenvalue one.
    Since $\Gamma(S_n,S_n/S_{\bm{x}_P},S')$ satisfies Eq.~\eqref{eq:radical_condition}, it is connected, and thus $S'$ is a generating set of $S_n$. Thus, by applying Lemma~\ref{lem:transition_Schreier_graph} with $G=S_n$, $H=S_{\bm{x}_P}$, and $S=S'$, we have
    \begin{equation}
        T =  \frac{1}{|{\cal M}_n|}\sum_{s\in S'}{\rm Ind}_{S_{\bm{x}_P}}^{S_n}(1_{\mathbb{C}})(s), \label{eq:trans_mat}
    \end{equation}
    and
    \begin{equation}
        P_{\bm{1}} = \frac{1}{|S_n|}\sum_{s\in S_n}{\rm Ind}_{S_{\bm{x}_P}}^{S_n}(1_{\mathbb{C}})(s). \label{eq:stabilized_point}
    \end{equation}
    Furthermore, since $T$ is symmetric, we have
    \begin{align}
        \sup_{\zeta\in\mathbb{C}[S_n/S_{\bm{x}_P}]:\,\zeta^{\dagger}\zeta \leq 1} \sqrt{\zeta^{\dagger} (T-P_{\bm{1}})^2 \zeta} &= \sup_{\zeta\in\mathbb{C}[S_n/S_{\bm{x}_P}]:\,\zeta^{\dagger}\zeta \leq 1} \zeta^{\dagger}|(T-P_{\bm{1}})|\zeta \\
        &=\lambda(\Gamma(S_n,S_n/S_{\bm{x}_P},S')),\label{eq:singular_value_bound}
    \end{align}
    where the last equality follows from Eq.~\eqref{eq:second_largest_expression} in Lemma~\ref{lem:transition_Schreier_graph}. Applying the fourth condition of the definition of the $(P, R, \delta)$-radical spectral expander in Def.~\ref{def:radical_spectral_expander} to Eq.~\eqref{eq:singular_value_bound}, we have 
    \begin{equation}
        \sup_{\zeta\in\mathbb{C}[S_n/S_{\bm{x}_P}]:\,\zeta^{\dagger}\zeta \leq 1} \sqrt{\zeta^{\dagger} (T-P_{\bm{1}})^2 \zeta} \leq e^{-\frac{nR}{2}},
    \end{equation}
    which proves Eq.~\eqref{eq:end_step_3}.

    \noindent{\bf Step 4:} 
    This step aims to show
    \begin{equation}
        \sup_{i\in\{1,\ldots,l\}}\sup_{\eta\in J_i:\braket{\eta|\eta}\leq 1} \Braket{(\widetilde{\cal M}_n - \mathbb{E}_{S_n})\restriction_{J_i}(\eta)|(\widetilde{\cal M}_n - \mathbb{E}_{S_n})\restriction_{J_i}(\eta)} \leq \sup_{\zeta\in \mathbb{C}[S_n/S_{\bm{x}_P}]:\zeta^{\dagger}\zeta\leq 1} \zeta^{\dagger} (T-P_{\bm{1}})^2\zeta. \label{eq:end_step_4}
    \end{equation}
    We consider intertwining the two representations $(J_i,{\cal V}_s\restriction_{J_i})$ for $i\in\{1,\ldots,l\}$ and ${\rm Ind}_{S_{\bm{x}_P}}^{S_n}(1_{\mathbb{C}})$ of $S_n$, where an intertwiner is introduced in Sec.~\ref{sec:representation_theory}.
    First, observe that each $J_i$ satisfies $J_i\cap{\cal K}\neq\{0\}$; i.e., it contains a nontrivial element in ${\cal K}$. In fact, if there exists $i\in\{1,\ldots,l\}$ such that $J_i\cap{\cal K}=\{0\}$, then ${\cal J} \cap J_i^{\perp}$ is also a subspace closed under $\{{\cal V}_s:s\in S_n\}$ that contains ${\cal K}$. But this contradicts with the fact that ${\cal J}$ is the smallest subspace closed under $\{{\cal V}_s:s\in S_n\}$ that contains ${\cal K}$.
    Now, for any $i\in\{1,\ldots,l\}$, an irreducible representation equivalent to $(J_i,{\cal V}_s\restriction_{J_i})$ is contained in $\mathrm{Ind}_{S_{\bm{x}_P}}^{S_n}(1_{\mathbb{C}})$ from Lemma~\ref{lem:containing_irreducibles} since $J_i\cap{\cal K}\neq\{0\}$ and ${\cal K}$ is a trivial representation subspace for $S_{\bm{x}_P}$. Then, from Lemma~\ref{lem:Schur}, there exists a surjective partial isometry $R_i:\mathbb{C}[S_n/S_{\bm{x}_P}]\to J_i$ onto the irreducible representation subspace $J_i$ such that
    \begin{equation}
        R_i\circ {\rm Ind}_{S_{\bm{x}_P}}^{S_n}(1_{\mathbb{C}})(s) = {\cal V}_s\restriction_{J_i}\circ R_i.
    \end{equation}
    Thus, from Eqs.~\eqref{eq:def_tilde_M_n} and \eqref{eq:trans_mat}, we have
    \begin{equation}
        R_i \circ T = \widetilde{\cal M}_n\restriction_{J_i}\circ R_i, \label{eq:intertwine_trans}
    \end{equation}
    and from Eqs.~\eqref{eq:subgroup_averaging} and \eqref{eq:stabilized_point}, we have
    \begin{equation}
        R_i \circ P_{\bm{1}} = \mathbb{E}_{S_n}\restriction_{J_i}\circ R_i. \label{eq:intertwine_stab}
    \end{equation}
    We thus have, for any $i\in\{1,\ldots,l\}$, 
    \begin{align}
        &\sup_{\eta\in J_i:\braket{\eta|\eta}\leq 1} \Braket{(\widetilde{\cal M}_n - \mathbb{E}_{S_n})\restriction_{J_i}(\eta)|(\widetilde{\cal M}_n - \mathbb{E}_{S_n})\restriction_{J_i}(\eta)} \nonumber \\
        & = \sup_{\zeta\in \mathbb{C}[S_n/S_{\bm{x}_P}]:\zeta^{\dagger}\zeta\leq 1} \Braket{(\widetilde{\cal M}_n - \mathbb{E}_{S_n})\restriction_{J_i}\circ R_i(\zeta)|(\widetilde{\cal M}_n - \mathbb{E}_{S_n})\restriction_{J_i}\circ R_i(\zeta)} \\
        &= \sup_{\zeta\in \mathbb{C}[S_n/S_{\bm{x}_P}]:\zeta^{\dagger}\zeta\leq 1} \Braket{R_i\circ(T-P_{\bm{1}})(\zeta)|R_i\circ(T-P_{\bm{1}})(\zeta)}, \label{eq:partial_isometry_remained}
    \end{align}
    where the first equality follows from the fact that $R_i$ is surjective, and the second equality follows from Eqs.~\eqref{eq:intertwine_trans} and~\eqref{eq:intertwine_stab}.
    Since a partial isometry does not increase a Hilbert-space norm, we have 
    \begin{align}
        \sup_{\zeta\in \mathbb{C}[S_n/S_{\bm{x}_P}]:\zeta^{\dagger}\zeta\leq 1} \Braket{R_i\circ(T-P_{\bm{1}})(\zeta)|R_i\circ(T-P_{\bm{1}})(\zeta)} \leq \sup_{\zeta\in \mathbb{C}[S_n/S_{\bm{x}_P}]:\zeta^{\dagger}\zeta\leq 1} \zeta^{\dagger} (T-P_{\bm{1}})^2\zeta.
    \end{align}
    Combining this with Eq.~\eqref{eq:partial_isometry_remained} and the fact that these inequalities hold for any $i\in\{1,\ldots,l\}$, we have Eq.~\eqref{eq:end_step_4}.

    Now, combining Eq.~\eqref{eq:modified_goal} in Step 2, Eq.~\eqref{eq:end_step_3} in Step~3, and Eq.~\eqref{eq:end_step_4} in Step~4, we conclude that Eq.~\eqref{eq:the_goal} in Step~1 holds, which implies the statement of the lemma.
\end{proof}

From the lemma proved above, we will prove the c-q channel resolvability.
We use the following lemma, which can be shown from the complex interpolation theory~\cite{Beigi2013,Dupuis2022,Beigi2023operator,Beigi2023} (see Theorem~4 in Ref.~\cite{Beigi2013}).
\begin{lemma}\label{lem:interpolation}
    For $\alpha\in[1,2]$, let $\tau_{B^n}$ be an operator on ${\cal H}_B^{\otimes n}$ satisfying $\|\tau_{B^n}\|_{\alpha}=1$.  Then, we have
    \begin{equation}
        \|\Theta_{S'}(\tau_{B^n})\|_{\alpha} \leq \sup_{\eta\in{\cal S}_p({\cal H}_B^{\otimes n}):\|\eta\|_1\leq 1} \|\Theta_{S'}(\eta)\|_1^{\frac{2}{\alpha}-1} \sup_{\eta\in{\cal S}_p({\cal H}_B^{\otimes n}):\|\eta\|_2\leq 1} \|\Theta_{S'}(\eta)\|_2^{2-\frac{2}{\alpha}}. \label{eq:interpolation}
    \end{equation}
\end{lemma}

Now, we are ready to prove the universal c-q-channel resolvability.

\begin{proof}[Proof of Theorem~\ref{theo:universal_resolvability}]
    The condition~\eqref{eq:asymptotic_rate_resolvability} is obvious from the definition of the $(P, R, \delta)$-radical spectral expander in Def.~\ref{def:radical_spectral_expander} and the construction of a sequence of codebooks in Lemma~\ref{lem:existence}.
    For Eq.~\eqref{eq:exponent_resolvability}, the quantity we need to evaluate can be rewritten using $\Theta_{S'}$ in Eq.~\eqref{eq:def_theta_map} as
    \begin{equation}
        \left\|\frac{1}{|{\cal M}_n|}\sum_{\bm{x}\in{\cal M}_n}W^{\otimes n}(\bm{x}) - \overline{W}^n_P \right\|_1 = \left\|\Theta_{S'}\bigl(W^{\otimes n}(\bm{x}_P)\bigr)\right\|_1.
    \end{equation}
    This indeed holds since ${\cal M}_n=S'(\bm{x}_P)$ from the first condition of Def.~\ref{def:radical_spectral_expander} and $S_n(\bm{x}_P)$ is a multiset with each element of ${\cal T}_P$ appearing exactly $|S_{\bm{x}_P}|$ times from Lemma~\ref{lem:bijection_coset_type}.
    In the following three steps, we prove, for any $\alpha\in[1,2]$,
    \begin{equation}
        \left\|\Theta_{S'}\bigl(W^{\otimes n}(\bm{x}_P)\bigr)\right\|_1 \leq 2^{\frac{2}{\alpha}-1}\exp\left[-\frac{\alpha - 1}{\alpha}n(R -\breve{I}_{\alpha}(X:B)_{W^P})\right], \label{eq:resolvability_alt}
    \end{equation}
    which immediately implies Eq.~\eqref{eq:exponent_resolvability}. Since the case $\alpha=1$ trivially holds, we actually prove the above inequality for $\alpha\in(1,2]$.

    \noindent {\bf Step 1:}
    In this step, we prove
    \begin{equation}
    \begin{split}
        \left\|\Theta_{S'}\bigl(W^{\otimes n}(\bm{x}_P)\bigr)\right\|_1 &\leq \inf_{\sigma_B:\cup_{x\in{\cal X}}\, {\rm supp}(W(x))\subseteq{\rm supp}(\sigma_B)}\|(\sigma_B^{\otimes n})^{-\frac{1}{2\alpha'}} W^{\otimes n}(\bm{x}_P) (\sigma_B^{\otimes n})^{-\frac{1}{2\alpha'}}\|_{\alpha}\\
        &\hspace{3cm}\sup_{\tau_{B^n}\in{\cal S}_{\alpha}({\cal H}_B^{\otimes n}):\|\tau_{B^n}\|_{\alpha}= 1}\|\Theta_{S'}(\tau_{B^n})\|_{\alpha}.
        \end{split}\label{eq:bound_by_alpha-norm}
    \end{equation}
    First, observe that $\alpha'\coloneqq \alpha/(\alpha-1)$ satisfies $2\leq\alpha' < \infty$ for $1<\alpha\leq 2$. Then, for any density operator $\sigma_B$ with $\cup_{x\in{\cal X}}\, {\rm supp}(W(x))\subseteq{\rm supp}(\sigma_B)$, we have
    \begin{equation}
        \Theta_{S'}\bigl(W^{\otimes n}(\bm{x}_P)\bigr) = (\sigma_B^{\otimes n})^{\frac{1}{2\alpha'}} (\sigma_B^{\otimes n})^{-\frac{1}{2\alpha'}} \Theta_{S'}\bigl(W^{\otimes n}(\bm{x}_P)\bigr)  (\sigma_B^{\otimes n})^{-\frac{1}{2\alpha'}}(\sigma_B^{\otimes n})^{\frac{1}{2\alpha'}}.
    \end{equation}
    From the H\"older's inequality $\|ABC\|_1\leq \|A\|_p\|B\|_q\|C\|_r$ with $\frac{1}{p}+\frac{1}{q}+\frac{1}{r}=1$, we have
    \begin{align}
        \left\|\Theta_{S'}\bigl(W^{\otimes n}(\bm{x}_P)\bigr)\right\|_1 &\leq \|(\sigma_B^{\otimes n})^{\frac{1}{2\alpha'}}\|_{2\alpha'} \left\|(\sigma_B^{\otimes n})^{-\frac{1}{2\alpha'}}\Theta_{S'}\bigl(W^{\otimes n}(\bm{x}_P)\bigr)(\sigma_B^{\otimes n})^{-\frac{1}{2\alpha'}}\right\|_{\alpha} \|(\sigma_B^{\otimes n})^{\frac{1}{2\alpha'}}\|_{2\alpha'}\\
        &\leq \left\|\Theta_{S'}\left((\sigma_B^{\otimes n})^{-\frac{1}{2\alpha'}} W^{\otimes n}(\bm{x}_P) (\sigma_B^{\otimes n})^{-\frac{1}{2\alpha'}}\right)\right\|_{\alpha}\\
        &\leq \|(\sigma_B^{\otimes n})^{-\frac{1}{2\alpha'}} W^{\otimes n}(\bm{x}_P) (\sigma_B^{\otimes n})^{-\frac{1}{2\alpha'}}\|_{\alpha} \sup_{\tau_{B^n}\in{\cal S}_{\alpha}({\cal H}_B^{\otimes n}):\|\tau_{B^n}\|_{\alpha}= 1}\|\Theta_{S'}(\tau_{B^n})\|_{\alpha}, 
    \end{align}
    where we used the permutation invariance ${\cal V}_s(\sigma_B^{\otimes n})=\sigma_B^{\otimes n}$ for any $s\in S_n$ as well as $\|\rho^{\frac{1}{p}}\|_{p}^p=\tr[(\rho^{\frac{1}{p}})^p]=1$ for any density operator $\rho$ in the second inequality. Taking minimization over $\sigma_B$ proves Eq.~\eqref{eq:bound_by_alpha-norm}.
    
    \noindent {\bf Step 2:}
    In this step, we prove
    \begin{equation}
        \inf_{\sigma_B:\cup_{x\in{\cal X}}\, {\rm supp}(W(x))\subseteq{\rm supp}(\sigma_B)}\|(\sigma_B^{\otimes n})^{-\frac{1}{2\alpha'}} W^{\otimes n}(\bm{x}_P) (\sigma_B^{\otimes n})^{-\frac{1}{2\alpha'}}\|_{\alpha} = \exp\!\left[\frac{\alpha - 1}{\alpha}n\,\breve{I}_{\alpha}(X:B)_{W^P}\right].\label{eq:bound_on_alpha-norm}.
    \end{equation}
    Since $W^n(\bm{x}_P)=W({x_P}_1)\otimes\cdots\otimes W({x_P}_n)$, we have
    \begin{align}
        \|(\sigma_B^{\otimes n})^{-\frac{1}{2\alpha'}} W^{\otimes n}(\bm{x}_P) (\sigma_B^{\otimes n})^{-\frac{1}{2\alpha'}}\|_{\alpha} &= \prod_{x\in{\cal X}} \left\|\sigma_B^{-\frac{1}{2\alpha'}}W(x) \sigma_B^{-\frac{1}{2\alpha'}}\right\|_{\alpha}^{n P(x)} \\
        &= \exp\!\left[\frac{\alpha - 1}{\alpha}n\sum_{x\in{\cal X}}P(x)\widetilde{D}_{\alpha}(W(x)\|\sigma_B)\right], \label{eq:bound_on_alpha-norm_middle}
    \end{align}
    where the sandwiched R\'enyi divergence $\widetilde{D}_{\alpha}(\rho\|\sigma)$ is defined in Eq.~\eqref{eq:sandwiched_renyi}.
    By minimizing over $\sigma_B$ on both side of Eq.~\eqref{eq:bound_on_alpha-norm_middle}, we have
    Eq.~\eqref{eq:bound_on_alpha-norm} from the definition of $\breve{I}_{\alpha}(X:B)_{W^P}$ in Eq.~\eqref{eq:def_augustin}.

    \noindent {\bf Step 3:} This step proves 
    \begin{equation}
         \sup_{\tau_{B^n}\in{\cal S}_{\alpha}({\cal H}_B^{\otimes n}):\|\tau_{B^n}\|_{\alpha}= 1}\|\Theta_{S'}(\tau_{B^n})\|_{\alpha} \leq 2^{\frac{2}{\alpha} - 1}e^{-\frac{\alpha - 1}{\alpha}nR}. \label{eq:channel_alpha_bound}
    \end{equation}
    From the definition of $\Theta_{S'}$ in Eq.~\eqref{eq:def_theta_map}, we observe that 
    \begin{equation}
        \|\Theta_{S'}(\eta)\|_1 \leq \frac{1}{|{\cal M}_n|}\sum_{s'\in S'}\frac{1}{|S_{\bm{x}_P}|}\sum_{s\in S_{\bm{x}_P}}\|V_{s'} V_{s} \eta V_{s}^{\dagger} V_{s'}^{\dagger} \|_1 + \frac{1}{|S_n|}\sum_{s\in S_n}\|V_s \eta V_s^{\dagger}\|_1 \leq 2\|\eta\|_1, \label{eq:one_norm_bound}
    \end{equation}
    where we used the triangle inequality in the first inequality and the unitary invariance of the norm in the second inequality.
    On the other hand, since the codebook ${\cal M}_n$ satisfies the $(P, R, \delta_2(n))$-radical spectral expander by assumption, we have 
    \begin{equation}
        \|\Theta_{S'}(\eta)\|_2 \leq e^{-\frac{nR}{2}}\|\eta\|_2,
    \end{equation}
    from Lemma~\ref{lem:two_norm_bound}. 
    Applying this and Eq.~\eqref{eq:one_norm_bound} to Eq.~\eqref{eq:interpolation} in Lemma~\ref{lem:interpolation}, we have Eq.~\eqref{eq:channel_alpha_bound}.
    
    Combining Eqs.~\eqref{eq:bound_by_alpha-norm}, \eqref{eq:bound_on_alpha-norm}, and \eqref{eq:channel_alpha_bound} proves Eq.~\eqref{eq:resolvability_alt}.  The last statement of the theorem follows from the fact that $\breve{I}_{\alpha}(X:B)_{W^P}$ is continuous and monotonically non-decreasing on $\alpha\in[1,2]$ and that $\breve{I}_{\alpha}(X:B)_{W^P}\to I(X:B)_{W^P}$ as $\alpha\to 1$.
\end{proof}

\begin{remark}
    The optimality of the error exponent in Eq.~\eqref{eq:exponent_resolvability} for universal resolvability coding is unknown as far as the authors' knowledge. It is at least as tight as the current best upper bound on the error exponent for the soft-covering lemma with a random constant-composition coding~\cite{Cheng2022error}.
\end{remark}

\section{Universal private channel coding} \label{sec:universal_private}
Combining the universal c-q channel coding studied in Secs.~\ref{sec:universal_c-q} and \ref{sec:universal_resolvability}, we can construct a universal private channel coding.  Let $W_{BE}:{\cal X}\to{\cal D}({\cal H}_B\otimes{\cal H}_E)$ be a c-q-q channel.  Input classical information is transmitted to a receiver with a channel $W_B:x \mapsto \tr_E[W_{BE}(x)]$ and to an eavesdropper with a channel $W_E:x\mapsto \tr_B[W_{BE}(x)]$.  A private channel coding aims to transmit a message from the sender to the receiver so that the message is unknown to the eavesdropper with an i.i.d.~use of the channel $W_{BE}^{\otimes n}$.  More explicitly, a private-channel code ${\cal C}_n$ is a quadruple $({\cal J}_n,K_n,\varphi_n^K,Y_n^K)$ of a message set ${\cal J}_n=\{1,\ldots,J_n\}$, an integer $K_n$, a stochastic encoder $\varphi_n^K;{\cal J}_n\to {\cal X}^n$ that maps each message $j\in{\cal J}_n$ to one of $K_n$ codewords uniformly randomly (the codebook thus has $|{\cal J}_n|\times K_n$ elements), and a decoder $Y_n^K:{\cal D}({\cal H}_B^{\otimes n})\to {\cal J}_n$ such that the following two conditions hold.
\begin{itemize}
\item (Correctness) The decoding error $P_{\rm err}({\cal C}_n, W_B)$ defined as
\begin{equation}
    P_{\rm err}({\cal C}_n,W_B) \coloneqq 1 - \frac{1}{|{\cal J}_n|}\sum_{j\in{\cal J}_n} \tr[W_B^{\otimes n}(\varphi_n^K(j)) \, Y_n^K(j)] \label{eq:correctness_def}
\end{equation}
decays as $n\to \infty$.
\item (Secrecy) Any message $j\in{\cal J}_n$ satisfies
\begin{equation}
    \left\|W_E^{\otimes n}(\varphi_n^{K}(j)) - \overline{W}_E^n(P)\right\|_1 \to 0,\label{eq:secrecy_asympt}
\end{equation}
where $\overline{W}_E^n(P)\coloneqq \frac{1}{|{\cal T}_P|}\sum_{\bm{x}\in{\cal T}_P}W_E^{\otimes n}(\bm{x})$
\end{itemize}
The first condition implies that the receiver successfully decodes the transmitted message with unit probability in the asymptotic limit, and the second implies that the transmitted message does not leak to Eve in the asymptotic limit.
The transmission rate $R$ and the randomness consumption rate $R_E$ for the code ${\cal C}_n$ are defined, respectively, as
\begin{align}
    R&=\lim_{n\to\infty} \frac{1}{n}\log|{\cal J}_n|, \\
    R_E&=\lim_{n\to\infty} \frac{1}{n}\log K_n.
\end{align}

Then, what we will show in this section is the following statement.
\begin{theorem} \label{theo:private_channel}
    Let $W_{BE}:{\cal X}\to{\cal D}({\cal H}_B\otimes{\cal H}_E)$ be a c-q-q channel.  Then, for any real numbers $R_B < I(X:B)_{W_B^P}$ and $R_E > I(X:E)_{W_E^P}$, there exists a sequence $\{{\cal C}_n\}_{n\in L}$ of codes ${\cal C}_n=({\cal J}_n, K_n, \phi_n^K, Y_n^K)$, where $L$ denotes the subsequence of $\mathbb{N}$, that satisfy the following conditions:
    \begin{align}
        \log (K_n|{\cal J}_n |) &= nR_B - \Omega(|{\cal X}|^2\log n),  \label{eq:size_total_codebook}\\
        \log K_n &= nR_E + \Omega(\log n), \label{eq:size_each_set}\\ 
        P_{\rm err}({\cal C}_n,W_B) &\leq \min_{\alpha\in[0,1]}\exp\!\left[-\frac{\alpha}{2(1+\alpha)}n\bigl(I_{1-\alpha}(X:B)_{W_B^P}-R_B\bigr) + \Omega\bigl(|{\cal X}|(d_B^2+|{\cal X}|)\log n\bigr)\right], \label{eq:correctness_cond}
    \end{align}  
    where $d_B\coloneqq \dim{\cal H}_B$,
    and for any $j\in{\cal J}_n$, 
    \begin{equation}
        \left\|W_E^{\otimes n}(\varphi_n^{K}(j)) - \overline{W}_E^n(P)\right\|_1 \leq \min_{\alpha\in[1,2]} \exp\left[-\frac{\alpha-1}{\alpha}n\bigl(R_E - \breve{I}_{\alpha}(X:E)_{W_E^P}\bigr)+ \Omega(1)\right]. \label{eq:secrecy_cond}
    \end{equation}
    Here, the code ${\cal C}_n$ is determined without the knowledge on $W_{BE}$ except for the values of $I(X:B)_{W_B^P}$ and $I(X:E)_{W_E^P}$. Thus, the transmission rate $I(X:B)-I(X:E)$ with the randomness-consumption rate $I(X:E)$ is achievable.
\end{theorem}

Unlike Ref.~\cite{Datta2010}, the above statement does not rely on the random-coding argument, and thus we can construct a sequence of codes completely independent of the channel.
The idea of a construction of a private channel code is the same as the conventional strategy; the sender Alice sends a pair $(m,k)$ of a message $m$ and a random number $k$ so that the receiver Bob can decode $m$ while the eavesdropper Eve cannot distinguish the message $m$ due to the randomness by $k$.  The universal c-q channel coding is used for Bob's successful decoding of $m$, and the universal resolvability coding is used to disturb Eve's information gain through $k$.  The nontrivial part is a construction of a sequence of codebooks that are $(P, R_B)$-good codebooks for $(m, k)$ and $(P, R_E)$-radical spectral expander for $k$.  In fact, we use a slightly weaker notion of $(P, R_B)$-good codebook since the ultimate goal for Bob is to decode a message $m$.
\begin{definition}[$(\alpha,\beta)$-setwise $(P, R, \delta)$-good codebook]\label{def:weak_good}
    For a type $P\in{\cal P_n}$, a positive real number $R<H(X)_P$, and a real number $\delta>0$, let $K_n$ and $J_n$ be integers, and let $p_{{\cal L}_n}$ for a multiset ${\cal L}_n$ with every element in ${\cal X}^n$ be as defined in Eq.~\eqref{eq:p_M}. Let $({\cal K}_n^1,\ldots,{\cal K}_n^{J_n})$ be a $J_n$-tuple of sets ${\cal K}_n^j$ of sequences in ${\cal X}^n$ with $|{\cal K}_n^j|=K_n$ for any $j\in{\cal J}_n$, and let ${\cal M}_n\coloneqq \uplus_{j=1}^{J_n}{\cal K}_n^j$. Then, the multiset ${\cal M}_n$ is said to be $(\alpha,\beta)$-setwise $(P, R, \delta)$-good codebook if it satisfies the following conditions.
    \begin{enumerate}
        \item Every element of ${\cal M}_n$ is in ${\cal T}_P$.
        \item The cardinality $|{\cal M}_n\setminus{\cal K}_n^j|=K_n(J_n-1)=\exp[n R - \delta]$.
        \item For any $j\in\{1,\ldots,J_n\}$, there exists a subset ${\cal L}_n^j\subseteq{\cal K}_n^j$ with $|{\cal L}_n^j|= \alpha K_n$ such that any $\bm{x}\in {\cal K}_n^j\setminus{\cal L}_n^j$ satisfies
        \begin{equation}
            |{\cal T}_{\bm{V}}(\bm{x})\cap ({\cal M}_n\setminus{\cal K}_n^j)|\leq \beta |{\cal T}_{\bm{V}}(\bm{x})| e^{-n(H(X)_P-R)}, \label{eq:beta_approximate}
        \end{equation}
        for any $\bm{V}\in{\cal V}(\bm{x},{\cal X})$.
    \end{enumerate}
\end{definition}
The above definition implies that ${\cal M}_n\setminus{\cal K}_n^j$ is the $(P, R, \delta)$-good set defined in Def.~\ref{def:good_set} for all the entries of ${\cal K}_n^j$ up to a factor $\beta$ except its $\alpha$ portion ${\cal L}_n^j$. Thus, the $(P, R, \delta)$-good codebook defined in Def.~\ref{def:P-R_good} is equivalent to $(1,1)$-setwise $(P, R, \delta)$-good codebook by taking each ${\cal K}_n^j$ as a singleton. As a corollary of Lemma~\ref{lem:reduction_to_iid}, we have the following.
\begin{corollary}\label{cor:weak_collision}
    Let ${\cal M}_n\coloneqq \uplus_{j=1}^{J_n}{\cal K}_n^j$ be an $(\alpha,\beta)$-setwise $(P, R, \delta)$-good codebook as defined in Def.~\ref{def:weak_good}. Let $p_{{\cal L}_n}$ be as defined in Eq.~\eqref{eq:p_M}. Then, for any $j\in\{1,\ldots,J_n\}$ and for any $\bm{x}\in{\cal K}_n^j\setminus{\cal L}_n^j$, we have
    \begin{equation}
        \forall \bm{y}\in{\cal T}_P\setminus\{\bm{x}\}, \qquad \frac{1}{|S_{\bm{x}}|}\sum_{s\in S_{\bm{x}}}p_{{\cal M}_n\setminus{\cal K}_n^j}(s(\bm{y})) \leq \beta e^{-nH(X)_P + \delta}.
    \end{equation}
\end{corollary}
\begin{proof}
    It is obvious by following the same proof as that of Lemma~\ref{lem:reduction_to_iid}.
\end{proof}
The following lemma states that there exists a sequence of codebooks, each of which satisfies both the $(P, R_E, \delta_2)$-spectral expander and the $(\alpha,\beta)$-setwise $(P, R_B, \delta'_1)$-good codebook through a random construction.
\begin{lemma} \label{lem:codebook_required}
    Let $P$ be a probability distribution and $R_B$ and $R_E$ are real numbers satisfying $R_E < R_B < H(X)_P$. Then, there exists a $J_n$-tuple $({\cal K}_n^1,\ldots,{\cal K}_n^{J_n})$ of sets of codewords in ${\cal X}^n$ with each $|{\cal K}_n^j|=K_n$ such that each ${\cal K}_n^j$ is a $(P, R_E, \delta_2(n))$-radical spectral expander with $\delta_2(n)\leq \log(12n\log|{\cal X}|)$ and the union $\biguplus_{j=1}^{J_n}{\cal K}_n^j\eqqcolon{\cal M}_n$ is a $(16\epsilon_n,\epsilon_n^{-1})$-setwise $(P, R_B, \delta'_1(n))$-good codebook with $\delta'_1(n)=\log(n+1)^{|{\cal X}|(|{\cal X}|+1)}$ for any $\epsilon_n>0$.
\end{lemma}
\begin{proof}
    {\bf Step 0:} We will show that with a random construction, a tuple $({\cal K}_n^1,\ldots,{\cal K}_n^{J_n})$ satisfying the required properties can be obtained with a probability strictly larger than zero.  Let $\zeta_1$ and $\zeta_2$ be constants satisfying $\zeta_1+\zeta_2 < 1$.  Let us randomly sample $(\log4) e^{n R_E}\log(2e^{nH(X)_P} e^{n\log|{\cal X}|} \zeta_1^{-1})$ elements from $S_n$ with replacement and define ${\cal S}_1$ as the set of sampled elements.  Define a multiset ${\cal K}^1_n={\cal S}_1\bm{x}_P\uplus {\cal S}^{-1}_1\bm{x}_P$, where ${\cal S}_1\bm{x}_P\coloneqq \{s(\bm{x}_P):s\in {\cal S}_1\}$, which leads to 
    \begin{equation}
        |{\cal K}^1_n|=(2\log4) e^{nR_E} \log(2e^{n(H(X)_P+\log|{\cal X}|)} \zeta_1^{-1}) \eqqcolon K_n. \label{eq:size_K_prime}
    \end{equation}
    Repeat this procedure $2J_n$ times to obtain ${\cal S}_j$ and ${\cal K}^j_n$ for all $j\in\{1,\ldots,2J_n\}$, where $J_n$ is defined to satisfy   
    \begin{equation}
        J_n\geq 2\quad \text{and} \quad K_n(J_n-1) = \exp[n R_B - \delta'_1(n)], \label{eq:size_of_M}
    \end{equation}
    Define $\overline{\cal M}_n\coloneqq \uplus_{j=1}^{2J_n}{\cal K}^j_n$.  We show the statement of the lemma with the following three steps. In Step~1, we show that, for all $j\in 2J_n$, ${\cal K}_n^j$ simultaneously satisfies the $(P, R_E, \delta_2(n))$-radical spectral expander with a probability no smaller than $1 - \zeta_1$. In Step~2, we show that any element $\bm{x}\in\overline{\cal M}_n$ satisfies Eq.~\eqref{eq:beta_approximate} for any $\bm{V}\in{\cal V}(\bm{x},{\cal X})$ and $\beta=\epsilon_n^{-1}$ with the probability no smaller than $4\epsilon_n$. In Step~3, we show that each ${\cal K}_n^j$ satisfies the third condition of $(\alpha,\beta)$-setwise $(P, R_B, \delta'_1(n))$-good codebook in Def.~\ref{def:weak_good} with $\alpha\beta={\cal O}(1)$ against $\overline{\cal M}_n$ with a probability no smaller than $\zeta_2 / 2$. In Step~4, by eliminating half of the $2J_n$-tuple, we show that the resulting $J_n$-tuple, ${\cal M}_n$, is $(\alpha,\beta)$-setwise $(P, R_B, \delta'_1(n))$-good codebook.

    \noindent {\bf Step~1:} This step aims to show
    \begin{equation}
        \mathrm{Pr}\bigl[{\cal K}_n^j\text{ satisfies $(P, R_E, \delta_2(n))$-radical spectral expander for any }j\in\{1,\ldots,2J_n\}\bigr] \geq 1 - \zeta_1. \label{eq:simultaneous_expander}
    \end{equation}
    From Lemma~\ref{lem:random_construction} and the construction of ${\cal S}_j$ given above, we have for any $j\in\{1,\ldots,2J_n\}$ that
    \begin{equation}
        \mathrm{Pr}\left[\lambda\Bigl(\Gamma(S_n,S_n/S_{\bm{x}_P},{\cal S}_j\uplus{\cal S}_j^{-1})\Bigr) \geq e^{-nR_E/2}\right] \leq \zeta_1 e^{-n\log|{\cal X}|}, \label{eq:individual_expander}
    \end{equation}
    where we used $|S_n/S_{\bm{x}_P}|=|{\cal T}_P|\leq e^{nH(X)_P}$ from Lemma~\ref{lem:bijection_coset_type} and Eq.~\eqref{eq:type_cardinality_bound}. When $n$ satisfies $2\zeta_1^{-1}\leq e^{n\log|{\cal X|}}$, we have
    \begin{equation}
        K_n= (2\log4)e^{nR_E} \log(2e^{n(H(X)_P+\log|{\cal X}|)} \zeta_1^{-1}) \leq 4e^{n R_E} \log(e^{3n \log|{\cal X}|}) = e^{n R_E + \log(12n\log|{\cal X}|)}. \label{eq:bound_cardinality}
    \end{equation}
    From Eqs.~\eqref{eq:individual_expander} and \eqref{eq:bound_cardinality}, the probability that ${\cal K}_n^j$ is a $(P, R_E, \delta_2(n))$-radical spectral expander with $\delta_2(n)\leq \log(12n\log|{\cal X}|)$ is no smaller than $1-\zeta_1 e^{-n\log|{\cal X}|}$. Since $2 J_n \leq e^{n\log|{\cal X}|}$ holds from Eq.~\eqref{eq:size_of_M} for a sufficiently large $n$, Eq.~\eqref{eq:simultaneous_expander} follows by applying the union bound.

    \noindent {\bf Step 2:} This step aims to show that any $\bm{x}\in \overline{\cal M}_n$ satisfies
    \begin{equation}
        \mathrm{Pr}\left[\sum_{\bm{V}\in{\cal V}(\bm{x},{\cal X})}\frac{\bigl|{\cal T}_{\bm{V}}(\bm{x})\cap (\overline{\cal M}_n\setminus {\cal K}^l_n)\bigr|}{|{\cal T}_{\bm{V}}(\bm{x})|}\geq \epsilon_n^{-1}e^{-n (H(X)_P - R_B)}\right] \leq 4\epsilon_n. \label{eq:high_probability_good}
    \end{equation}
    Consider an arbitrary element $\bm{x}\in{\cal K}^l_n$ for an arbitrary $l\in\{1,\ldots,2J_n\}$. Since any element of $(\uplus_{j\in\{1,\ldots,2J_n\}\setminus\{l\}} {\cal S}_j(\bm{x}_P))$ is uniformly randomly chosen from ${\cal T}_P$ and independently of $\bm{x}$, the probability that each element of $(\uplus_{j\in\{1,\ldots,2J_n\}\setminus\{l\}} {\cal S}_j(\bm{x}_P))$ is contained in ${\cal T}_{\bm{V}}(\bm{x})$ is given by $|{\cal T}_{\bm{V}}(\bm{x})|/|{\cal T}_P|$ for any $\bm{V}\in{\cal V}(x,{\cal X})$.
    Thus, regarding all the elements of $\uplus_{j\in \{1,\ldots,2J_n\}\setminus\{l\}} {\cal S}_j$ as random variables, we have
    \begin{align}
        \mathbb{E}_{\uplus_{j\in\{1,\ldots,2J_n\}\setminus\{l\}} {\cal S}_j} \left[\frac{|{\cal T}_{\bm{V}}(\bm{x})\cap \bigl(\uplus_{j\in\{1,\ldots,2J_n\}\setminus\{l\}} {\cal S}_j(\bm{x}_P)\bigr)|}{|{\cal T}_{\bm{V}}(\bm{x})|}\right] &= \frac{|{\cal T}_{\bm{V}}(\bm{x})|}{|{\cal T}_P||{\cal T}_{\bm{V}}(\bm{x})|} \frac{K_n}{2} (2J_n-1) \\
        &\leq K_n J_n (n+1)^{|{\cal X}|}e^{-nH(X)_P}, \label{eq:expectation_S_j}
    \end{align}
    for any $\bm{V}\in {\cal V}(\bm{x},{\cal X})$, where we used $|{\cal S}_j|=K_n/2$ in the equality and Eq.~\eqref{eq:type_cardinality_bound} in the inequality.  Since $S_n \ni s\mapsto s^{-1} \in S_n$ is a one-to-one map, the replacement of ${\cal S}_j$ with ${\cal S}_j^{-1}$ implies
    \begin{align}
        \mathbb{E}_{\uplus_{j\in\{1,\ldots,2J_n\}\setminus\{l\}} {\cal S}_j} \left[\frac{\bigl|{\cal T}_{\bm{V}}(\bm{x})\cap \bigl(\uplus_{j\in \{1,\ldots,2J_n\}\setminus\{l\}} {\cal S}^{-1}_j(\bm{x}_P)\bigr)\bigr|}{|{\cal T}_{\bm{V}}(\bm{x})|}\right] \leq K_nJ_n(n+1)^{|{\cal X}|}e^{-nH(X)_P} \label{eq:expectation_S_inverse_j}
    \end{align}
    The combination of Eqs.~\eqref{eq:expectation_S_j} and \eqref{eq:expectation_S_inverse_j} implies
    \begin{equation}
        \mathbb{E}_{\uplus_{j\in\{1,\ldots,2J_n\}\setminus\{l\}} {\cal S}_j} \left[\frac{\bigl|{\cal T}_{\bm{V}}(\bm{x})\cap (\overline{\cal M}_n\setminus {\cal K}^l_n)\bigr|}{|{\cal T}_{\bm{V}}(\bm{x})|}\right] \leq 2K_n J_n (n+1)^{|{\cal X}|}e^{-nH(X)_P}.
    \end{equation}
    Using the fact that $|{\cal V}(\bm{x},{\cal X})|\leq (n+1)^{|{\cal X}|^2}$ as in Eq.~\eqref{eq:cardinality_conditonal_type}, we have
    \begin{align}
        \mathbb{E}_{\uplus_{j\in\{1,\ldots,2J_n\}\setminus\{l\}} {\cal S}_j} \left[\sum_{\bm{V}\in{\cal V}(\bm{x},{\cal X})}\frac{\bigl|{\cal T}_{\bm{V}}(\bm{x})\cap (\overline{\cal M}_n\setminus {\cal K}^l_n)\bigr|}{|{\cal T}_{\bm{V}}(\bm{x})|}\right] &\leq 2K_n J_n (n+1)^{|{\cal X}|(|{\cal X}|+1)}e^{-nH(X)_P} \\
        &\leq 4e^{-n(H(X)_P  - R_B)},
    \end{align}
    where we used $J_n/(J_n-1)\leq 2$ from Eq.~\eqref{eq:size_of_M}.
    Since the left-hand side essentially counts the number of randomly generated elements of $S_n$ that are in the set, one can apply Markov's inequality to have
    \begin{equation}
        \mathrm{Pr}\left[\sum_{\bm{V}\in{\cal V}(\bm{x},{\cal X})}\frac{\bigl|{\cal T}_{\bm{V}}(\bm{x})\cap (\overline{\cal M}_n\setminus {\cal K}^l_n)\bigr|}{|{\cal T}_{\bm{V}}(\bm{x})|}\geq \epsilon_n^{-1}e^{-n (H(X)_P - R_B)}\right] \leq 4\epsilon_n.
    \end{equation}
    This directly implies Eq.~\eqref{eq:high_probability_good}.
    
    \noindent{\bf Step 3:}
    For any $l\in\{1,\ldots,2 J_n\}$, let Condition G for $\bm{x}\in{\cal K}_n^l$ be defined as 
    \begin{equation}
        \text{Condition G:}\qquad \forall\bm{V}\in{\cal V}(\bm{x},{\cal X}),\quad \frac{\bigl|{\cal T}_{\bm{V}}(\bm{x})\cap (\overline{\cal M}_n\setminus {\cal K}^l_n)\bigr|}{|{\cal T}_{\bm{V}}(\bm{x})|}\leq \epsilon_n^{-1}e^{-n (H(X)_P - R_B)}, \label{eq:condition_okay}
    \end{equation}
    Then, this step aims to show 
    \begin{equation}
        \mathrm{Pr}\bigl[\exists {\cal L}_n^l\subseteq{\cal K}_n^l \text{ such that } |{\cal L}^l_n|=\frac{8\epsilon_n}{\zeta_2}K_n \text{ and }\forall \bm{x}\in{\cal K}_n^l\setminus{\cal L}_n^l \text{ satisfies Condition G} \bigr] \geq 1-\frac{\zeta_2}{2}. \label{eq:existence_L}
    \end{equation}
    For any $k\in\{1,\ldots,K_n\}$, let $X^l_k$ be a random variable that takes the value zero if the $k$-th element $\bm{x}_k\in{\cal K}_n^l$ satisfies Condition G and one otherwise.
    Then, $\sum_{k=1}^{K_n}X^l_k$ corresponds to the number of elements in ${\cal K}_n^l$ that violates Condition G.
    Thus, Eq.~\eqref{eq:high_probability_good} yields
    \begin{equation}
        \mathbb{E}_{\uplus_{j\in\{1,\ldots,2J_n\}\setminus\{l\}} {\cal S}_j}\left[\sum_{k=1}^{K_n} X^l_k\right] \leq 4\epsilon_n K_n
    \end{equation}
    Since $X^l_k$ is a binary random variable, we apply Markov's inequality to have
    \begin{equation}
        \mathrm{Pr}\left[\sum_{k=1}^{K_n}X^l_k \geq \frac{8\epsilon_n}{\zeta_2}K_n\right] \leq \frac{\zeta_2}{2}, \label{eq:each_K_goal}
    \end{equation}
    which immediately implies Eq.~\eqref{eq:existence_L}.

    \noindent {\bf Step 4:}
    This step aims to show
    \begin{equation}
        \begin{split}
        &\mathrm{Pr}\Bigl[\exists{\cal M}_n\subset\overline{\cal M}_n \text{ is $(\epsilon_n\zeta_2/8,\epsilon_n^{-1})$-setwise $(P, R_B, \delta'_1(n))$-good codebook} \\
        &\qquad \text{with }|{\cal M}_n|=K_nJ_n \text{ and }\delta'_1(n)=|{\cal X}|(|{\cal X}|+1)\log(n+1) \Bigr] \geq 1 - \zeta_2.
        \end{split}\label{eq:setwise_good_high_probability}
    \end{equation}
    For any $l\in\{1,\ldots,2J_n\}$, let Condition H be defined for ${\cal K}_n^l$ as
    \begin{equation}
        \text{Condition H:}\quad \exists {\cal L}_n^l\subseteq{\cal K}_n^l \text{ such that } |{\cal L}^l_n|=\frac{8\epsilon_n}{\zeta_2}K_n \text{ and }\forall \bm{x}\in{\cal K}_n^l\setminus{\cal L}_n^l \text{ satisfies Condition G}.
    \end{equation}
    Let $Z_l$ be a random variable that takes the value zero if ${\cal K}_n^l$ satisfies Condition H and one otherwise. Then, $\sum_{l=1}^{2J_n}Z_l$ corresponds to the number of sets that violate Condition H. From Eq.~\eqref{eq:existence_L}, we have that
    \begin{equation}
        \mathbb{E}\left[\sum_{l=1}^{2J_n}Z_l\right] \leq \zeta_2 J_n.
    \end{equation}
    Since each $Z_l$ is a binary random variable, we can apply Markov's inequality to have
    \begin{equation}
        \mathrm{Pr}\left[\sum_{l=1}^{2J_n}Z_l \geq J_n\right] \leq \zeta_2.
    \end{equation}
    This means that with a probability no smaller than $\zeta_2$, at least $J_n$ sets satisfy Condition H. Thus, we prove the statement Eq.~\eqref{eq:setwise_good_high_probability}.

    \noindent Finally, we set $\zeta_2=1/2$ and $\zeta_1$ be a small constant satisfying $\zeta_1<1/2$. Then, from Step~1 and Step~4 and the union bound, the claimed random construction gives the $J_n$-tuple $({\cal K}^1_n,\ldots,{\cal K}_n^{J_n})$ that satisfies the claimed property with a probability no smaller than $1/2-\zeta_1$. One then obtains the claimed tuple of sets by repeating until success.
\end{proof}

With the construction of a codebook given in Lemma~\ref{lem:codebook_required}, we can show Theorem~\ref{theo:private_channel} as follows.
\begin{proof}[Proof of Theorem~\ref{theo:private_channel}]
    {\bf Step 0:} For a sufficiently large $n$, we take a $J_n$-tuple $({\cal K}_n^1,\ldots,{\cal K}_n^{J_n})$ of sets of codewords as given in Lemma~\ref{lem:codebook_required}, and set ${\cal M}_n\coloneqq \uplus_{j=1}^{J_n} {\cal K}_n^j$. Then, Eqs.~\eqref{eq:size_total_codebook} and \eqref{eq:size_each_set} are automatically satisfied. We define the stochastic encoder $\varphi_n^{K}$ as a map from $j$ to an element of ${\cal K}_n^j$ uniformly randomly. Then, we have
    \begin{equation}
        W_{BE}^{\otimes n}(\varphi_n^K(j)) = \frac{1}{K_n}\sum_{\bm{x}\in{\cal K}_n^j}W^{\otimes n}_{BE}(\bm{x}). \label{eq:average_codeword}
    \end{equation}
    Furthermore, we define the decoder $Y^K_n$ as 
    \begin{equation}
        Y^K_n(j) = \sum_{\bm{x}\in{\cal K}_n^j\setminus{\cal L}_n^j} \tilde{Y}_n(\bm{x}), \label{eq:sum_of_decoder}
    \end{equation}
    where $\tilde{Y}_n(\bm{x})$ is as 
    \begin{equation}
        \tilde{Y}_n(\bm{x}) = \left(\sum_{j\in{\cal J}_n}\sum_{\bm{y}\in{\cal K}_n^j\setminus{\cal L}_n^j}\Pi(\bm{y})\right)^{-\frac{1}{2}} \Pi(\bm{x})\left(\sum_{j\in{\cal J}_n}\sum_{\bm{y}\in{\cal K}_n^j\setminus{\cal L}_n^j}\Pi(\bm{y})\right)^{-\frac{1}{2}} ,
    \end{equation}
    with $\Pi(\bm{x})$ defined in Eq.~\eqref{eq:likelihood_ratio_test}.
    With these definitions of the stochastic encoder and the decoder, we prove Eqs.~\eqref{eq:secrecy_cond} and \eqref{eq:correctness_cond}. In fact, Eq.~\eqref{eq:secrecy_cond} simply follows from the fact that, for sufficiently large $n$, ${\cal K}_n^j$ satisfies the $(P, R_E, \delta_2(n))$-radical spectral expander for any $j\in{\cal J}_n$ by construction, which leads to
    \begin{equation}
        \left\|W_E^{\otimes n}(\varphi_n^{K}(j)) - \overline{W}_E^n(P)\right\|_1 \leq \min_{\alpha\in[1, 2]}\exp\left[-\frac{\alpha - 1}{\alpha}n\bigl(R_E- \breve{I}_{\alpha}(X:E)_{W_E^P}\bigr) + \frac{2-\alpha}{\alpha}\log 2\right],
    \end{equation}
    from Eq.~\eqref{eq:resolvability_alt} in Theorem~\ref{theo:universal_resolvability}. This proves Eq.~\eqref{eq:secrecy_cond}.
    
    It remains to show that Eq.~\eqref{eq:correctness_cond} is satisfied when ${\cal M}_n\coloneqq \uplus_{j\in{\cal J}_n} {\cal K}_n^j$ satisfies the $(16\epsilon_n,\epsilon_n^{-1})$-setwise $(P, R_B, \delta'_1(n))$-good codebook property as given in Lemma~\ref{lem:codebook_required} with an appropriate value of $\epsilon_n$, 
    where the average decoding error $P_{\rm err}({\cal C}_n,W_B)$ is given by
    \begin{align}
        P_{\rm err}({\cal C}_n,W_B) &= \frac{1}{J_n}\sum_{j=1}^{J_n}\tr\left[\frac{1}{K_n}\sum_{\bm{x}\in{\cal K}^j_n}W_B^{\otimes n}(\bm{x}) \bigl(I-Y_n^{K}(j)\bigr)\right]. 
    \end{align}
    We split the proof of this statement into the following five steps.

    \noindent {\bf Step 1:}
    This step splits the error evaluation of bad codewords, i.e., the codewords that are included in $\uplus_{j=1}^{J_n}{\cal L}_n^j$, from the others.
    The definition of $(16\epsilon_n,\epsilon_n^{-1})$-setwise $(P, R_B, \delta'_1(n))$-good codebook property implies
    \begin{align}
        P_{\rm err}({\cal C}_n,W_B) &= \frac{1}{J_n}\sum_{j=1}^{J_n}\tr\left[\frac{1}{K_n}\left(\sum_{\bm{x}\in{\cal K}^j_n\setminus{\cal L}_n^j}W_B^{\otimes n}(\bm{x}) \bigl(I-Y_n^{K}(j)\bigr) + \sum_{\bm{x}\in{\cal L}_n^j}W_B^{\otimes n}(\bm{x}) \bigl(I-Y_n^{K}(j)\bigr)\right)\right] \\
        & \leq \frac{1}{J_n}\sum_{j=1}^{J_n}\left(\tr\left[\frac{1}{K_n}\sum_{\bm{x}\in{\cal K}^j_n\setminus{\cal L}_n^j}W_B^{\otimes n}(\bm{x}) \bigl(I-Y_n^{K}(j)\bigr)\right] + \frac{|{\cal L}_n^j|}{ K_n}\right) \\
        &= \frac{1}{J_n}\sum_{j=1}^{J_n}\tr\left[\frac{1}{K_n}\sum_{\bm{x}\in{\cal K}^j_n\setminus{\cal L}_n^j}W_B^{\otimes n}(\bm{x}) \bigl(I-Y_n^{K}(j)\bigr)\right] + 16\epsilon_n. \label{eq:separate_bad}
    \end{align}
    \noindent {\bf Step 2:} This step aims to split the first term of Eq.~\eqref{eq:separate_bad} into the type I and the type II errors.
    Applying Hayashi-Nagaoka inequality~\cite{Hayashi2002} to $I-Y_n^K(j)$, we have
    \begin{equation}
        I - Y_n^K(j) = I - \sum_{\bm{x}\in{\cal K}_n^j\setminus{\cal L}_n^j}\tilde{Y}_n(\bm{x}) \leq 2 \left(I - \sum_{\bm{x}\in{\cal K}_n^j\setminus{\cal L}_n^j}\Pi(\bm{x}) \right) + 4 \sum_{l\in{\cal J}_n\setminus\{j\}}\sum_{{\cal K}_n^l\setminus{\cal L}_n^l}\Pi(\bm{x}),
    \end{equation}
    where $\Pi(\bm{x})$ is defined in Eq.~\eqref{eq:likelihood_ratio_test}.
    We thus have
    \begin{align}
        &\frac{1}{J_n}\sum_{j=1}^{J_n}\tr\left[\frac{1}{K_n}\sum_{\bm{x}\in{\cal K}^j_n\setminus{\cal L}_n^j}W_B^{\otimes n}(\bm{x}) \bigl(I-Y_n^{K}(j)\bigr)\right] \nonumber \\
        & \leq  \frac{1}{M_n}\sum_{j\in{\cal J}_n}\tr\left[\sum_{\bm{x}\in{\cal K}^j_n\setminus{\cal L}_n^j}W_B^{\otimes n}(\bm{x}) \left(2 \biggl(I - \sum_{\bm{y}\in{\cal K}_n^j\setminus{\cal L}_n^j}\Pi(\bm{y}) \biggr) + 4 \sum_{\bm{y}\in{\cal K}_n^l\setminus{\cal L}_n^l}\Pi(\bm{y})\right)\right] \\
        \begin{split}
        & \leq \frac{2}{M_n} \sum_{j\in{\cal J}_n}\tr\left[\sum_{\bm{x}\in{\cal K}^j_n\setminus{\cal L}_n^j}W_B^{\otimes n}(\bm{x})\bigl(I-\Pi(\bm{x})\bigr)\right] \\
        & \qquad + \frac{4}{M_n} \sum_{j\in{\cal J}_n}\tr\left[\sum_{\bm{x}\in{\cal K}^j_n\setminus{\cal L}_n^j}W_B^{\otimes n}(\bm{x}) \sum_{l\in{\cal J}_n\setminus\{j\}}\sum_{\bm{y}\in{\cal K}_n^l\setminus{\cal L}_n^l}\Pi(\bm{y})\right],
        \end{split}\label{eq:Hayashi_Nagaoka}
    \end{align}
    where we used $I - \sum_{\bm{y}\in{\cal K}_n^j\setminus{\cal L}_n^j}\Pi(\bm{y}) \leq I - \Pi(\bm{x})$ for any $\bm{x}\in{\cal K}_n^j\setminus{\cal L}_n^j$ in the second inequality. 

    \noindent {\bf Step 3:}
    This step aims to evaluate the first term of Eq.~\eqref{eq:Hayashi_Nagaoka}, so-called the type-I error.
    In Eq.~(21) of Ref.~\cite{Hayashi2009}, it is shown that the following holds for any $\alpha\in[0,1]$:
    \begin{equation}
        \tr\left[W_B^{\otimes n}(\bm{x})\bigl(I-\Pi(\bm{x})\bigr)\right] \leq (n+1)^{d+\frac{\alpha|{\cal X}|(d+2)(d-1)}{2}}C_n^{\alpha} e^{-n\alpha I_{1-\alpha}(X:B)_{W_B^P}}, 
        \label{eq:first_term_bound}
    \end{equation}
    which then leads to 
    \begin{equation}
        \frac{2}{M_n} \sum_{j\in{\cal J}_n}\tr\left[\sum_{\bm{x}\in{\cal K}^j_n\setminus{\cal L}_n^j}W_B^{\otimes n}(\bm{x})\bigl(I-\Pi(\bm{x})\bigr)\right] \leq 2(n+1)^{d+\frac{\alpha|{\cal X}|(d+2)(d-1)}{2}}C_n^{\alpha} e^{-n\alpha I_{1-\alpha}(X:B)_{W_B^P}} .
        \label{eq:type_I_bounded}
    \end{equation}

    \noindent {\bf Step 4:} The remaining task is to evaluate the second term in Eq.~\eqref{eq:type_I_bounded}, so called the type-II error. Let us reorder the summation as
    \begin{align}
    \begin{split}
        &\frac{4}{M_n}\sum_{j\in{\cal J}_n}\tr\left[\sum_{\bm{x}\in{\cal K}^j_n\setminus{\cal L}_n^j}W_B^{\otimes n}(\bm{x}) \sum_{l\in{\cal J}_n\setminus\{j\}}\sum_{\bm{y}\in{\cal K}_n^l\setminus{\cal L}_n^l}\Pi(\bm{y})\right] \\
        &= 4\sum_{j\in{\cal J}_n} \sum_{\bm{y}\in{\cal K}_n^j\setminus{\cal L}_n^j}\tr\left[\Pi(\bm{y})\,\frac{1}{M_n}\sum_{l\in{\cal J}_n\setminus\{j\}}\sum_{\bm{x}\in{\cal K}_n^l\setminus{\cal L}_n^l}W_B^{\otimes n}(\bm{x}) \right]. 
        \end{split}\label{eq:reordered}   
    \end{align}
    Using the probability distribution $p_{{\cal L}_n}$ defined in Eq.~\eqref{eq:p_M}, we can upper bound the summation over $W_B^{\otimes n}(\bm{x})$ as 
    \begin{align}
        \frac{1}{M_n}\sum_{l\in{\cal J}_n\setminus\{j\}}\sum_{\bm{x}\in{\cal K}_n^l\setminus{\cal L}_n^l}W_B^{\otimes n}(\bm{x}) &\leq \frac{1}{M_n}\sum_{l\in{\cal J}_n\setminus\{j\}}\sum_{\bm{x}\in{\cal K}_n^l}W_B^{\otimes n}(\bm{x})\\
        &\leq \sum_{\bm{x}\in{\cal T}_P\setminus{\cal K}_n^j}p_{{\cal M}_n\setminus{\cal K}_n^j}(\bm{x})\, W_B^{\otimes n}(\bm{x}),
    \end{align}
    which holds for any $j\in{\cal J}_n$. By inserting the above inequality to Eq.~\eqref{eq:reordered} and randomly permute the operators in the trace by $S_{\bm{y}}$, we have
    \begin{align}
        &\frac{4}{M_n}\sum_{j\in{\cal J}_n}\tr\left[\sum_{\bm{x}\in{\cal K}^j_n\setminus{\cal L}_n^j}W_B^{\otimes n}(\bm{x}) \sum_{l\in{\cal J}_n\setminus\{j\}}\sum_{\bm{y}\in{\cal K}_n^l\setminus{\cal L}_n^l}\Pi(\bm{y})\right] \nonumber \\
        &\leq 4\sum_{j\in{\cal J}_n} \sum_{\bm{y}\in{\cal K}_n^j\setminus{\cal L}_n^j}\tr\left[\Pi(\bm{y})\sum_{\bm{x}\in{\cal T}_P\setminus{\cal K}_n^j}p_{{\cal M}_n\setminus{\cal K}_n^j}(\bm{x})\, W_B^{\otimes n}(\bm{x}) \right] \\
        &=4\sum_{j\in{\cal J}_n} \sum_{\bm{y}\in{\cal K}_n^j\setminus{\cal L}_n^j}\frac{1}{|S_{\bm{y}}|}\sum_{s\in S_{\bm{y}}}\tr\left[\Pi(\bm{y}) \sum_{\bm{x}\in{\cal T}_P\setminus\{\bm{y}\}}p_{{\cal M}_n\setminus{\cal K}_n^j}(\bm{x})\,V_s W_B^{\otimes n}(\bm{x}) V_s^{\dagger}\right] \\
        &= 4\sum_{j\in{\cal J}_n} \sum_{\bm{y}\in{\cal K}_n^j\setminus{\cal L}_n^j}\frac{1}{|S_{\bm{y}}|}\sum_{s'\in S_{\bm{y}}}\tr\left[\Pi(\bm{y}) \sum_{\bm{x}\in{\cal T}_P\setminus\{\bm{y}\}}p_{{\cal M}_n\setminus{\cal K}_n^j}(s'(\bm{x}))\, W_B^{\otimes n}(\bm{x}) \right], \label{eq:before_applying_Lemma}
    \end{align}
    where we used in the first equality that $\Pi(\bm{y})$ is invariant under the adjoint action of $V_s$ for any $s\in S_{\bm{y}}$, and we used in the last equality that the set ${\cal T}_P\setminus\{\bm{y}\}$ is invariant under the action of $s\in S_{\bm{y}}$.
    By applying Corollary~\ref{cor:weak_collision}, we have
    \begin{align}
        &4\sum_{j\in{\cal J}_n} \sum_{\bm{y}\in{\cal K}_n^j\setminus{\cal L}_n^j}\frac{1}{|S_{\bm{y}}|}\sum_{s'\in S_{\bm{y}}}\tr\left[\Pi(\bm{y}) \sum_{\bm{x}\in{\cal T}_P\setminus\{\bm{y}\}}p_{{\cal M}_n\setminus{\cal K}_n^j}(s'(\bm{x}))\, W_B^{\otimes n}(\bm{x}) \right] \\
        &\leq 4\sum_{j\in{\cal J}_n} \sum_{\bm{y}\in{\cal K}_n^j\setminus{\cal L}_n^j} \tr\left[\Pi(\bm{y}) \sum_{\bm{x}\in{\cal T}_P\setminus\{\bm{y}\}}\epsilon_n^{-1}e^{-nH(X)_P + \delta'_1(n)}W_B^{\otimes n}(\bm{x})\right] \\
        &\leq 4\epsilon_n^{-1}e^{\delta'_1(n)}\sum_{\bm{y}\in{\cal M}_n}\tr\left[\Pi(\bm{y})\biggl(\sum_{x\in{\cal X}}P(x)W_B(x)\biggr)^{\otimes n}\right] \\
        &\leq 4\epsilon_n^{-1}e^{\delta'_1(n)}(n+1)^{\frac{(d+2)(d-1)}{2}}\sum_{\bm{y}\in{\cal M}_n}\tr\left[\Pi(\bm{y})\sigma_{U,n}\right], \label{eq:replacement_by_universal_symmetric}
    \end{align}
    where we used Eqs.~\eqref{eq:bound_universal_symmetric} and \eqref{eq:quantum_coefficient_bound} in the last inequality. Recalling the definition of $\Pi(\bm{x})$ in Eq.~\eqref{eq:likelihood_ratio_test}, we have
    \begin{equation}
        \tr[\Pi(\bm{y})\sigma_{U,n}] \leq \tr[\Pi(\bm{y}) C_n^{-1}\sigma_{\bm{y}}] \leq C_n^{-1}. \label{eq:constant_bound}
    \end{equation}
    Combining Eqs.~\eqref{eq:before_applying_Lemma}, \eqref{eq:replacement_by_universal_symmetric}, and \eqref{eq:constant_bound}, we have an upper bound on the type II error as
    \begin{align}
        \frac{4}{M_n}\sum_{j\in{\cal J}_n}\tr\left[\sum_{\bm{x}\in{\cal K}^j_n\setminus{\cal L}_n^j}W_B^{\otimes n}(\bm{x}) \sum_{l\in{\cal J}_n\setminus\{j\}}\sum_{\bm{y}\in{\cal K}_n^l\setminus{\cal L}_n^l}\Pi(\bm{y})\right] &\leq 4 K_n J_n\epsilon_n^{-1}e^{\delta'_1(n)}(n+1)^{\frac{(d+2)(d-1)}{2}}C_n^{-1} \\
        &\leq 8(n+1)^{\frac{(d+2)(d-1)}{2}} e^{n R_B}C_n^{-1}, \label{eq:type_II_bounded}
    \end{align}
    where we used Eq.~\eqref{eq:size_of_M} and $J_n/(J_n-1)\leq 2$ from the construction of the codebook.

    \noindent {\bf Step 5:}
    We finally prove Eq.~\eqref{eq:correctness_cond} in this step. From Eqs.~\eqref{eq:separate_bad}, \eqref{eq:type_I_bounded}, and \eqref{eq:type_II_bounded}, we have
    \begin{equation}
        \begin{split}
            P_{\rm err}({\cal C}_n,W_B) &\leq 2e^{[2d+\alpha|{\cal X}|(d+2)(d-1)]\log(n)/2}C_n^{\alpha} e^{-n\alpha I_{1-\alpha}(X:B)_{W_B^P}} \\
            &\qquad + 8\epsilon_n^{-1} e^{nR_B} e^{\delta'_1(n)+(d+2)(d-1)\log(n+1)/2} C_n^{-1}+16\epsilon_n.
        \end{split}
    \end{equation}
    By setting free parameters $C_n$ and $\epsilon_n$ as 
    \begin{align}
        C_n&=e^{\frac{n}{1+\alpha} (\alpha I_{1-\alpha}(X:B)_{W_B^P} + R_B)},\\
        \epsilon_n&= e^{\frac{\alpha n}{2(1+\alpha)}(I_{1-\alpha}(X:B)_{W_B^P} - R_B)},
    \end{align}
    we have
    \begin{equation}
        P_{\rm err}({\cal C}_n,W_B) \leq \min_{\alpha\in[0,1]} \exp\left[-\frac{\alpha}{2(1+\alpha)}n\bigl(I_{1-\alpha}(X:B)_{W_B^P} - R_B\bigr) + \Omega(|{\cal X}|(d_B^2 + |{\cal X}|)\log n\right],
    \end{equation}
    which proves Eq.~\eqref{eq:correctness_cond}.
\end{proof}

\section{Discussion}
In this work, we have developed a universal c-q channel resolvability coding and a fully universal c-q private channel coding.
The relevant property of a codebook that enables channel resolvability can be written in terms of an expander graph. Even though an expander graph was extensively studied in the field of error-correcting codes~\cite{Sipser1996}, it was not well studied in the context of information-theoretic security or resolvability coding. The tasks of these two fields are opposite in their objectives: the former aims to distinguish channel outputs, while the latter aims to simulate their statistical behavior.
Thus, our results shed new light on the interplay between coding theory and graph-theoretic structures in quantum information beyond universal channel resolvability.
Notice also that our universal resolvability theorem, Theorem~\ref{theo:universal_resolvability}, is even independent of the dimension of the channel output, unlike c-q channel coding in Theorem~\ref{theo:new_error_exponent}. Our universal resolvability theorem achieves strong universality in this sense.

There are several open problems. Are there any hints to the optimal error exponents for universal c-q channel coding and universal c-q channel resolvability coding? We have only derived achievable bounds, which lead to lower bounds on the error exponents, but upper bounds are unknown. Furthermore, (strong) converse exponents for these tasks are unknown as well. We leave these problems to future work.

\begin{acknowledgments}
    T.~M.~was supported by JST PRESTO, Grant Number JPMJPR24FA, Japan and JST CREST, Grant Number JPMJCR23I3, Japan.
    M.~H.~was supported in part by the National Natural Science Foundation of China under Grant 62171212, and the General R\&D Projects of 1+1+1 CUHK-CUHK(SZ)-GDST Joint Collaboration Fund (Grant No.~GRDP2025-022).
\end{acknowledgments}

\bibliography{universal_private}
\end{document}